\documentclass[11pt]{article}
\usepackage{CJK}
\usepackage[latin1]{inputenc}
\usepackage[T1]{fontenc}
\usepackage{mathptmx}
\usepackage{url}
\usepackage{latexsym}
\usepackage{cite}
\usepackage{multirow}
\usepackage{xcolor}
\usepackage{comment}
\usepackage[all]{xy}
\usepackage{mathdots}
\usepackage{setspace}
\usepackage{geometry}
\usepackage{algorithm}
\usepackage{algorithmic}
\usepackage{amsmath, amsthm, amsfonts, amssymb, mathrsfs, amsmath, float}
\usepackage{rotating}
\usepackage{bbm}
\usepackage{xypic}
\usepackage{graphicx}
\usepackage{epstopdf}

\newtheorem{myDef}{Definition}[section]
\newtheorem{prop}[myDef]{Proposition}
\newtheorem{lemma}[myDef]{Lemma}
\newtheorem{remark}[myDef]{Remark}
\newtheorem{theorem}[myDef]{Theorem}
\newtheorem{coro}[myDef]{Corollary}
\newtheorem{example}[myDef]{Example}

\def\Z{\mathbb{Z}}
\def\Zp{\mathbb{Z}_{(p)}}
\def\hatZp{\widehat{\mathbb{Z}}_{(p)}}
\def\Zx{\mathbb{Z}[x ]}
\def\hatZpanglex{\hatZp\langle x\rangle}
\def\Zpx{\mathbb{Z}_{(p)}[x]}
\def\Q{\mathbb{Q}}
\def\Qx{\mathbb{Q}[x]}
\def\N{\mathbb{N}}

\def\C{\mathbb{C}}
\def\sol{{\rm Sol}}
\def\rank{{\rm rank}}
\def\height{{\rm height}}
\def\CM{{\rm CMAT}}
\def\hf{{\rm PHNF}}
\def\lc{{\bf LC}}

\def\lt{{\bf LT}}
\def\gb{Gr{\"o}bner basis}
\def\Div{{\rm Divide}}

\def\f{{\bf{f}}}
\def\g{{\bf{g}}}
\def\r{{\bf{r}}}

\def\u{{\bf{u}}}
\def\v{{\bf{v}}}
\def\e{{\bf{e}}}
\def\0{{\bf{0}}}
\def\Syz{{\bf{Syz}}}

\def\X{\mathbb{X}}
\def\fb{{\mathbbm{f}}}

\def\gHNF{{\rm\,GHNF\,}}
\def\gHNFs{{\rm\,GHNFs\,}}

\def\cref#1{(\ref{#1})}

\begin{document}
\begin{CJK*}{GBK}{times new roman}

\title{A Polynomial-time Algorithm to Compute
 Generalized Hermite Normal Forms of Matrices over $\Z[x]$}
\author{Rui-Juan Jing, Chun-Ming Yuan, Xiao-Shan Gao\thanks{Corresponding Author.}\\
 KLMM,  Academy of Mathematics and Systems Science\\
 Chinese Academy of Sciences, Beijing 100190, China\\
 rjing@amss.ac.cn, \{cmyuan,xgao\}@mmrc.iss.ac.cn}
\date{}

\maketitle

\begin{abstract}
\noindent
In this paper, we give the first polynomial time algorithm to compute the
generalized Hermite normal form for a matrix $F$ over $\Zx$,
or equivalently, the reduced Gr\"obner basis of the $\Zx$-module
generated by the column vectors of $F$.
The algorithm has polynomial bit size  computational complexities
and is also shown to be practically more efficient than existing algorithms.
The algorithm is based on three key ingredients.
First, an F4 style algorithm to compute the Gr\"obner basis
is adopted, where a novel prolongation is designed such
that the sizes of coefficient matrices under consideration are nicely controlled.
Second, the complexity bound of the algorithm is achieved
by a nice estimation for the degree and height bounds
of the polynomials in the generalized Hermite normal form.
Third, fast algorithms to compute Hermite normal forms of
matrices over $\Z$ are used as the computational tool.

\end{abstract}


\section{Introduction}

The Hermite normal form (abbr. HNF) is a standard representation for matrices over principal ideal domains such as $\Z$ and $\Q[x]$, which has many applications in
algebraic group theory,  integer programming, lattices, linear Diophantine equations,
system theory, and analysis of cryptosystems \cite{cohen1993course,kannan1979polynomial,Micciancio2001}.
Efficient algorithms to compute HNF have been studied extensively until recently
\cite{cohen1993course,kannan1979polynomial,domich1987hermite,
storjohann1996asymptotically,Storjohann2013,beckermann2006normal,kaltofen1987fast,
Mulders2003377}.
Note that $\Zx$ is not a PID and
a matrix over $\Zx$ cannot be reduced to an HNF.
In \cite{gao2014binomial}, the concept of
generalized Hermite normal form (abbr. GHNF) is introduced and
it is shown that any matrix over $\Zx$ can be reduced to a GHNF.
Furthermore, a matrix $F=[\f_1,\ldots,\f_s]\in\Zx^{n\times s}$
is a GHNF if and only if the set of its column vectors $\fb=\{\f_1,\ldots,\f_s\}$
forms a reduced Gr\"obner basis of the $\Zx$-module generated by
$\fb$  in $\Z[x]^n$ under certain monomial order.
Similar to the concept of lattice \cite{cohen1993course},
a $\Zx$-module in $\Z[x]^n$ is called a $\Zx$-lattice which
plays the same role as lattice does in the study of
binomial ideals and toric varieties \cite{cox-toric}.
For instance, the decision algorithms for
some of the major properties of Laurent binomial difference ideals
and toric difference varieties are based on the computation of GHNFs
of the exponent matrices of the difference ideals \cite{gao2014binomial}.
This motivates the study of efficient algorithms to compute the
GHNFs.

The reduced Gr\"obner basis for a $\Zx$-lattice can be computed
with the Gr\"obner basis methods for modules over rings \cite{cox2005using,kandri1988computing,Lichtblau2003}.
However, such general algorithms do not take advatage of the special
properties of $\Zx$-modules and do not have a complexity analysis.
Also note that the worst case complexity of computing Gr\"obner bases in
$\Q[x_1,\ldots,x_n]$ is double exponential \cite{MM}.

The main contribution of this paper is to give an
algorithm to compute the GHNF of a matrix
$F\in\Zx^{n\times s}$ or the reduced Gr\"obner basis of the $\Zx$-lattice
generated by the column vectors of $F$,
which is both practically efficient and has
polynomial bit size computational complexity.
The algorithm consists of three main ingredients.

The first ingredient comes from the powerful idea in Faug{\`e}re's F4
algorithm~\cite{faugere1999new} and the XL algorithm~\cite{xl1} of Courtois et al.
To compute the Gr\"obner basis
of the ideal generated by $p_1,\ldots,p_m\in\Q[x_1,\ldots,x_n]$,
these algorithms apply efficient elimination algorithms from linear
algebra to the coefficient matrix  of $x_j^kp_i$ for certain $k$.
Although the F4 algorithm can not improve the worst case complexity,
it is generally faster than the classical Buchberger algorithm \cite{buchberger2006bruno}.
In this paper, to compute the GHNF of $F=[\f_1,\ldots,\f_s]\in\Zx^{n\times s}$ with columns $\f_i$,
due to the special structure of the Gr\"obner bases
in $\Zx$, we design a novel method to do certain prolongations $x^k\f_i$
such that the sizes of the coefficient matrices of those $x^k\f_i$
are nicely controlled.

The second ingredient is a nice estimation for the degree and height bounds
of the polynomials in the GHNF $G\in\Zx^{n\times s}$ of $F\in\Zx^{n\times m}$.
We show that the degrees and the heights of the key elements of $G$ are bounded by $nd$
and $6n^3d^2(h+1+\log(n^2d))$, respectively,
where $d$ and $h$ are the maximal degree and maximal height of the polynomials in $F$, respectively.
Furthermore, we show that $G=FU$ for a matrix $U\in\Zx^{m\times s}$
and the degrees of the polynomials in $U$ are bounded by a polynomial
in $n,d,h$, which is a key factor in the complexity  analysis of our algorithm.
Note that the degree bound also depends on the coefficients of $F$.
The bounds about the GHNF are obtained based on the powerful methods introduced
by Aschenbrenner in~\cite{aschenbrenner2004ideal}, where the first  double
exponential algorithm for the ideal membership problem in $\Z[x_1,\ldots, x_n]$
is given.
In order to find the degree and height bounds for the GHNF,
we need to find solutions
of linear equations over $\Zx$, whose degree and height are bounded.
Due to the special structure of the Gr\"obner basis in $\Zx$,
we give better bounds than those in~\cite{aschenbrenner2004ideal}.

The third ingredient is to use efficient algorithms to compute the HNF
for matrices over $\Z$.
The computationally dominant step of our algorithm is to compute the
HNF of the  coefficient matrices of those prolongations $x^k\f_i$
obtained in the first ingredient.
The first polynomial-time algorithm to compute HNF was given by Kannan and Bachem \cite{kannan1979polynomial}
and there exist many efficient algorithms to compute
HNFs for matrices over $\Z$ \cite{domich1987hermite,storjohann1996asymptotically,Storjohann2013,cohen1993course}
and matrices over $\Q[x]$ \cite{beckermann2006normal,kaltofen1987fast,Mulders2003377}.
Note that the GHNF for a matrix over $\Zx$
cannot be recovered from its HNF over $\Q[x]$ directly.
In the complexity analysis of our algorithm,
we use the HNF algorithm with the best bit size complexity bound \cite{Storjohann2013}.

The algorithm is implemented in Magma and Maple and their default HNF commands are used in our implementation.
In the case of $\Zx$,
our algorithm is shown to be more efficient than the Gr\"obner basis
algorithm in Magma and Maple.
In the general case, the proposed algorithm
is also very efficient in  that quite large problems can be solved.

The rest of this paper is organized as follows.
In Section 2, we introduce several notations for Gr{\"o}bner bases of $\Zx$ lattices.
In Section 3, we give degree and height bounds for the GHNF of a matrix over $\Zx$.
In Section 4, we give the algorithm  to compute the GHNF and analyze its complexity.
Experimental results are shown in Section 5.
Finally, conclusions are presented in Section 6.

\section{Preliminaries}
In this section, some basic notations and properties about Gr{\"o}bner bases for $\Zx$ lattices will be given. For more details, please refer to \cite{cox2005using,gao2014binomial,kandri1988computing}.

For brevity, a $\Zx$ module in $\Zx^n$ is called a {\em $\Zx$ lattice}.
Any $\Zx$ lattice $L$ has a finite set of generators
$\{{\bf f}_1,\ldots,{\bf f}_s\}\subset \Zx^n$ and this fact is denoted as
$L=({\bf f}_1,\ldots,{\bf f}_s)_{\Zx}$.
%
If ${\bf f}_i = [f_{1,i},\ldots,f_{n,i}]^\tau$, then we call $M=[\f_1,\ldots,\f_s]=[f_{i,j}]_{n\times s}$
a \emph{matrix representation} of $L=({\bf f}_1,\ldots,{\bf f}_s)_{\Zx}$.
If $n=1$, $M$ is called a \emph{polynomial vector}.

A monomial ${\bf m}$ in $\Zx^n$ is an element of the form $x^k{\bf e}_i\in \Zx^n$,
where  $k\in \mathbb{N}$, and ${\bf e}_i$ is the canonical $i$-th unit vector in $\Zx^n$.
A term  in $\Zx^n$ is a multiplication of an integer $a\in \Z$ and a monomial ${\bf m}$, that is $a{\bf m}$.
The admissible order $\prec$ on monomials in $\Zx^n$ can be defined naturally: $x^{\alpha}{\bf e}_i\prec x^{\beta}{\bf e}_j$ if
$i<j$ or $i=j$ and $\alpha < \beta$.
The order $\prec$ can be naturally extended to terms: $ax^{\alpha}{\bf e}_i\prec bx^{\beta}{\bf e}_j$ if and only if $x^{\alpha}{\bf e}_i\prec x^{\beta}{\bf e}_j$ or $i=j ,~\alpha = \beta~{\rm and} ~|a|<|b|$.

With the admissible order $\prec$, any ${\bf f}\in \Zx^n$ can be written in a unique way as a $\Z$-linear combination of monomials,
\begin{center}
  ${\bf f}=\sum_{i=1}^{s}c_i{\bf m}_i$,
\end{center}
where $c_i \neq 0$ and ${\bf m}_1\prec {\bf m}_2\prec\cdots\prec {\bf m}_s$. We define the \emph{leading coefficient, leading monomial}, and \emph{leading term} of ${\bf f}$ as
${\bf LC}({\bf f})=c_s$, ${\bf LM}({\bf f})={\bf m}_s$, and ${\bf LT}({\bf f})=c_s{\bf m}_s$, respectively.

The order $\prec$ can be extended to elements of $\Zx^n$ in a natural way: for ${\bf f,g}\in \Zx^n,{\bf f\prec g}$ if and only if ${\bf LT(f)\prec LT(g)}$.
We will use the order $\prec $  throughout this paper.

For two terms $ax^{\alpha}{\bf e}_i$ and $bx^{\beta}{\bf e}_j$ in $\Zx^n$ with $b\ne0$,
$ax^{\alpha}{\bf e}_i$ is called $\{bx^{\beta}{\bf e}_j\}$\emph{-reduced}
if one of the following conditions is valid: $i\ne j$; $i=j$ and $\alpha < \beta$;
or $i=j$, $\alpha \ge \beta$, and $0\le a < |b|$.
For any $\f,\g\in \Zx^n$ with $\g\ne0$, $\f$ is called $\g$-reduced if any term of $\f$ is $\lt(\g)$-reduced.
If $\f$ is not $\g$-reduced, then by the reduction algorithm for the polynomials in $\Z[x]$~\cite{Lichtblau2003},
one can compute a unique $\r$ and a quotient $q\in \Zx$  such that
$\r=\f- q\g$ is $\g$-reduced and is denoted as $\r=\overline{\f}^{\g}$.
If $\f$ is $\g$-reduced, then set $\overline{\f}^{\g}$ to be $\f$.
For ${\bf f}\in \Zx^n$ and $G=[{\bf g}_1,\ldots,{\bf g}_m]\in \Zx^{n\times m}$ with ${\bf g}_1\prec \ldots \prec {\bf g}_m$, $\f$ is called $G$-reduced if any term of $\f$ is $\lt(\g_i)$-reduced for $i=1,\ldots,m$.
Let ${\bf r}_{m+1}=\f$ and for $i=m,m-1,\ldots,1$, set
$\r_i = \overline{\r_{i+1}}^{\g_i}.$
Denote $\r_1=\overline{\f}^{G}$ and say $\f$ is reduced to $\r_1$ by $G$.

\begin{myDef}
Let ${\bf f,g}\in \Zx^n,~{\bf LT(f)}=ax^k{\bf e}_i,~{\bf LT(g)}=bx^s{\bf e}_j$, $s\leq k$. Then the S-vector of ${\bf f}$ and ${\bf g}$ is defined as follows: if $i \neq j$ then $S({\bf f,g})=\0$; otherwise

\begin{equation}\label{eq-svec}
S({\bf f,g})=\left\{
  \begin{array}{ll}
    {\bf f}-\frac{a}{b}x^{k-s}{\bf g}, & \hbox{if $b|a$;} \\
    \frac{b}{a}{\bf f}-x^{k-s}{\bf g}, & \hbox{if $a|b$;} \\
    u{\bf f}+vx^{k-s}{\bf g}, & \hbox{if $a\nmid b\hbox{ and }~b\nmid a,~\emph{where}~ \gcd(a,b)=ua+vb$.}
  \end{array}
\right.
\end{equation}
If $n=1$, the S-vector is called S-polynomial, which is the same with the definition in \cite{kandri1988computing}.
\end{myDef}
\begin{myDef}
\label{GBdef}
A finite set $G\subset \Zx^n$ is called a Gr{\"o}bner basis for the $\Zx$ lattice $L$ generated by $G$ if for any $\f\in L$, there exists $\g\in G$, such that  $\lt(\g)|\lt(\f)$. A Gr{\"o}bner basis $G$ is called reduced if for any ${\bf g}\in G,~{\bf g}$ is $G\setminus \{{\bf g}\}$-reduced.
A Gr{\"o}bner basis $G$ is called minimal if for any ${\bf g}\in G,~\lt(\g)$ is $G\setminus \{{\bf g}\}$-reduced.
\end{myDef}
It is easy to see that $G$ is a Gr{\"o}bner basis if and only if  $\overline{\g}^G=0$ for any
$\g\in(G)_{\Zx}$.
The Buchberger criterion for Gr\"obner basis is still true:
$G$ is a Gr{\"o}bner basis if and only if $\overline{S({\bf f},{\bf g})}^{G}={\bf 0}$ for all ${\bf f},{\bf g}\in G$.
Gr{\"o}bner bases in this paper are assumed to be ranked in an increasing order with respect to the admissible order $\prec $. That is, if $G=\{{\bf g}_1,\ldots,{\bf g}_s\}$ is a Gr{\"o}bner basis, then $\g_1\prec \ldots\prec \g_s$. To make the reduced Gr\"obner basis unique, we further assume that
$\lc(\g_i)>0$ for any   $\g_i\in G$.

We need the following property for  Gr{\"o}bner bases in $\Zx$.
\begin{prop}[\cite{gao2014binomial}]
\label{GB form}
Let $B=\{b_1,\ldots,b_k\}$ be the reduced Gr{\"o}bner basis of a $\Zx$ module in $\Zx$, $b_1\prec \cdots\prec b_k$, and ${\bf LT}(b_i)=c_ix^{d_i}\in \mathbb{N}[x]$. Then
\begin{enumerate}
  \item $0\leq d_1<\cdots<d_k$.
  \item $c_k|\cdots|c_1$ and $c_i\neq c_{i+1}$ for $1\leq i \leq k-1$.
  \item $\frac{c_i}{c_k}|b_i$ for $1\leq i<k$. Moreover, if $\widetilde{b_1}$ is the primitive part of $b_1$, then $\widetilde{b_1}|b_i$, for $1<i\leq k$.
\end{enumerate}
\end{prop}
This proposition also applies to the minimal Gr{\"o}bner bases.
Here are three Gr\"{o}bner bases in $\Z[x]$\hbox{\rm:} $\{ 2, x \}$,
$ \{ 12, 6x+6, 3x^2+3x, x^3+x^2 \}$, $ \{ 9x+3, 3x^2+4x+1 \}$.

For a polynomial set $F=\{f_1,\ldots,f_m\}$ in $\Zx$, we denote by ${\rm Content}(F)$ the GCD of the contents of $f_i$ and ${\rm Primpart}(F)=\gcd(F)/{\rm Content}(F)$ the primitive part of $F$.
Now, we give a refined description of Gr{\"o}bner bases for ideals in $\Zx$.
\begin{prop}[\cite{lazard1985grobner}]
\label{GB form'}
$G=\{g_1,\ldots, g_n\}$ with $\deg(g_1)<\cdots<\deg(g_n)$ is a minimal Gr{\"o}bner basis of $(f_1,\ldots, f_m)$ in $\Zx$ if and only if
$g_1=ab_1\cdots b_{n-1}\tilde{g_1}$, $g_n = a h_{n}\tilde{g_1}$,
and $g_i=a b_i\cdots b_{n-1}h_i\tilde{g_1}, 2\leq i\leq n-1$,
such that
 \begin{enumerate}
 \item[{\rm \romannumeral 1)}] $a={\rm Content}(f_1,\ldots,f_m)$;
 \item[{\rm \romannumeral 2)}]$\tilde{g}_1={\rm Primpart}(f_1,\ldots,f_m)$;
 \item[{\rm \romannumeral 3)}]  $h_i \in \Zx$ is monic with degree $d_i$, and $0<d_2<\cdots<d_n$;
 \item[{\rm \romannumeral 4)}] $b_i\in \Z, b_i\neq \pm1$, and $h_{i+1}\in (h_i, b_{i-1}h_{i-1},\ldots,b_2\ldots b_{i-1}h_2, b_1\ldots b_{i-1})$, for $1\leq i\leq n-1$, where $h_1=1$.
\end{enumerate}
\end{prop}

Next, we introduce the concept of generalized Hermite normal form. Let
\begin{equation}\label{reducedGB}
  \mathcal{C}=\left(
     \begin{smallmatrix}
       c_{11} & \ldots & c_{1,l_1} & c_{1,l_1+1} & \ldots & \ldots & \ldots & \ldots & \ldots & \ldots & \ldots \\
       \ldots & \ldots & \ldots & \ldots & \ldots & \ldots & \ldots & \ldots & \ldots & \ldots & \ldots \\
       c_{r_1,1} & \ldots & c_{r_1,l_1} & c_{r_1,l_1+1} & \ldots & \ldots & \ldots & \ldots & \ldots & \ldots & \ldots \\
       0 & \ldots & 0 & c_{r_1+1,1} & \ldots & c_{r_1+1,l_2} & \ldots & \ldots & \ldots & \ldots & \ldots \\
      \ldots & \ldots & \ldots & \ldots & \ldots & \ldots & \ldots & \ldots & \ldots & \ldots & \ldots \\
       0 & \ldots & 0 & c_{r_2,1} & \ldots & c_{r_2,l_2} & \ldots & \ldots & \ldots & \ldots & \ldots \\
      \ldots & \ldots & \ldots & \ldots & \ldots & \ldots & \ldots & \ldots & \ldots & \ldots & \ldots \\
       0 & \ldots & 0 & 0 & \ldots & 0 & \ldots & 0 & c_{r_{t-1}+1,1} & \ldots & c_{r_{t-1}+1,l_t} \\
       \ldots & \ldots & \ldots & \ldots & \ldots & \ldots & \ldots & \ldots & \ldots & \ldots & \ldots \\
       0 & \ldots & 0 & 0 & \ldots & 0 & \ldots & 0 & c_{r_t,1} & \ldots & c_{r_t,l_t} \\
     \end{smallmatrix}
   \right)_{n\times m}
   \end{equation}
whose elements are in $\Zx$. It is clear that $n=r_t$ and $m=\sum_{i=1}^{t}l_i$.
Assume
\begin{center}
  $c_{i,j}=c_{i,j,0}x^{d_{ij}}+\cdots+c_{i,j,d_{ij}}$,
\end{center}
and assume  $c_{i,j,0}\ge0$.
Then the leading term of ${\bf c}_{r_i,j}$ is $c_{r_i,j,0}x^{d_{r_i,j}}{\bf e}_{r_i}$,
where ${\bf c}_{r_i,j}$ is the $(l_1+\cdots+l_{i-1}+j)$-th column of $\mathcal{C}$.

\begin{myDef}\label{definition1}
The matrix $\mathcal{C}$ is called a generalized Hermite normal form (abbr. \gHNF) if it satisfies the following conditions:
\begin{description}
  \item[1)] $0\leq d_{r_i,1}<d_{r_i,2}<\cdots<d_{r_i,l_i}$ for any $i$.
  \item[2)] $c_{r_i,l_i,0}|\ldots |c_{r_i,2,0}|c_{r_i,1,0}$.
  \item[3)] $S({\bf c}_{r_i,j_1},{\bf c}_{r_i,j_2})=x^{d_{r_i,j_2}-d_{r_i,j_1}}{\bf c}_{r_i,j_1}-\dfrac{c_{r_i,j_1,0}}{c_{r_i,j_2,0}}{\bf c}_{r_i,j_2}$ can be reduced to zero by the column vectors of the matrix for any $1\leq i\leq t,~1\leq j_1<j_2 \leq l_i$.
  \item[4)] ${\bf c}_{r_i,j}$ is reduced with respect to the column vectors of the matrix other than ${\bf c}_{r_i,j}$, for any $1 \leq i\leq t, 1\leq j \leq l_i$.
\end{description}
\end{myDef}

\begin{theorem}[\cite{gao2014binomial}]\label{th-gg}
$\{{\bf f}_1,\ldots,{\bf f}_s\}\subset \Zx^n$ is a reduced Gr{\"o}bner basis under
the monomial order $\prec$ and ${\bf f}_1\prec {\bf f}_2\prec \ldots\prec {\bf f}_s$ if and only if the polynomial matrix $[{\bf f}_1,\ldots,{\bf f}_s]$ is a \gHNF.
\end{theorem}

\section{Degree and height bounds for the GHNF}\label{boundsanalysis}

We first give some notations.
Let $f\in R[x]$, where $R$ is a subring of $\C$.
Denote by $|f|$ the maximal absolute value of the coefficients of $f$.
Let $\height(f)=\log|f|$, with $\height(0)=0$.
For $F=\{f_1,\ldots,f_m\}\subset R[x]$,
let $\deg(F)=\max_{1\le i\le m}\deg(f_i)$ and
$\height(F)=\max_{1\le i\le m}\height(f_i)$.

For a prime $p\in \Z$,  let $\Z_{(p)}$ be the local ring of $\Z$ at $(p)$.  For $a = u p^t\in \Z$ where $u$ is a unit in $\Z_{(p)}$, let $v_p(a)=t$ be the $p$-adic valuation.
Let $\hatZp$ be the completion\cite{aschenbrenner2004ideal,eisenbud1995} of $\Zp$
 and $\hatZp[x]$ the polynomial ring with coefficients in $\hatZp$.
 Denote by $\hatZpanglex$ the completion of $\hatZp[x]$\cite{aschenbrenner2004ideal,eisenbud1995}.

For any subring $R$ of $\C$ or $\hatZp$ and $\f_1,\ldots,\f_s$ in $R[x]^n$,
let $(\f_1,\ldots,\f_s)_{R[x]}$ be the $R[x]$ module generated by
$\f_1,\ldots,\f_s$ in $R[x]^n$.

\subsection{Degree and height bounds in $\Z[x]$}
In this section, we give several basic degree and height bounds in
$\Z[x]$.
By the extended Euclidean algorithm, we have
\begin{lemma}\label{gcd<d}
Let $k$ be a field, $f_1,\ldots,f_m\in k[x]$, and $d=\max_{1\le i\le m}\deg(f_i)$. Then there exist $g_1,\ldots,g_m\in k[x]$ with $\deg(g_i)<d$ for any $i$, satisfying
  $\gcd(f_1,\ldots,f_m)=f_1g_1+\cdots+f_mg_m$.
\end{lemma}

In this section, we assume $f_1,\ldots,f_m\in \Z[x],~d=\max_{1\le i\le m}\deg(f_i)$, and  $h=\height(f_1,\ldots,f_m)$, unless specified otherwise explicitly.

\begin{lemma}\label{heightboundQx}
If $1\in (f_1,\ldots,f_m)_{\Q[x]}$,
  then $\delta=f_1g_1+\cdots+f_mg_m$~ for
  some $\delta\in \Z\setminus\{0\}$ with $\height(\delta)\le d(2h+\log (d+1))$ and some $g_1,\ldots,g_m\in \Z[x]$ with degree $< d$ .
  In this case, the height of the \gHNF of $[f_1,\ldots,f_m]$ is $\le  d(2h+\log (d+1))$.
 \end{lemma}
\begin{proof}
By Lemma \ref{gcd<d}, we have $1=f_1u_1+\cdots+f_mu_m$, where $u_i\in \Q[x]$ of degree $< d$. Assume $f_i=a_{i0}+\cdots+a_{id}x^d$,~$u_j=b_{j0}+\cdots+b_{j,d-1}x^{d-1}$.
Then we have the matrix equation $Ab=[1,0,\ldots,0]^\tau\in\Z^{2d}$, where
 $A=[A_1,\ldots,A_m]$ with
$$ A_i=\left(
           \begin{array}{cccc}
             a_{i0} &  &  &  \\
             a_{i1} & a_{i0} &  &  \\
             \vdots &  & \ddots &  \\
             a_{i,d} &  &  & a_{i0} \\
              & \ddots &  & \vdots \\
              & &  & a_{i,d} \\
           \end{array}
         \right)_{2d\times d}
         $$
 for $i=1,\ldots,m$, and $b=[b_{1,0},\ldots,b_{1,d-1},\ldots,b_{m,0},\ldots,b_{m,d-1}]^\tau\in \Q^{md}$. Let $t=\rank(A)\le 2d$.
By the Cramer's rule, $\delta$ can be bounded by the nonzero $t\times t$ minors of $A$.
 By the Hadamard's inequality, we have $0<\delta \le ((d+1)a^2)^d$, where $a=\max_{i,j}|a_{ij}|$. So $\height(\delta)\le d(2h+\log (d+1))$.
In this case, $\delta\in(f_1,\ldots, f_m)_{\Zx}$. Hence, the height of \gHNF of $[f_1,\ldots,f_m]$ is $\le \height(\delta)$.
\end{proof}

The following lemma is given by Gel'fond \cite{gelfond1960} and a simpler proof can be found in \cite[p178]{zippel1993effective}.
\begin{lemma}\label{lm-gcdheight}
Let $P_1$ and $P_2$ be two monic polynomials in $\mathbb{C}[x]$, such that $\deg(P_1) + \deg(P_2) = d $. Then
$ |P_1||P_2| \le (d+1)^{1/2}2^d|P_1P_2|.$
\end{lemma}

The following lemma gives a height bound for the gcd in $\Zx$.
\begin{lemma}\label{gcdheightbound}
Let $f_1,\ldots,f_m\in\Z[x]$ and $g=\gcd(f_1,\ldots,f_m)$ in $\Z[x]$. Then the height of $g$ is bounded by $\frac{1}{2}\log(d+1)+d\log 2+h$.
\end{lemma}
\begin{proof}
  Since $g=\gcd(f_1,\ldots,f_m)$ is in $\Z[x]$, for each $i=1,\ldots,m$, there exists a $g_i\in \Z[x]$ such that $gg_i=f_i.$
  Let $g'=g/{\bf LC}(g)$ and $g_i'=g_i/{\bf LC}(g_i)$. Then $f_i'=f_i/{\bf LC}(f_i)=f_i/{\bf LC}(g){\bf LC}(g_i)$ and $|f_i| = |f_i'||{\bf LC}(f_i)|$.
  Let $d_i=\deg(f_i)$.
  By Lemma \ref{lm-gcdheight}, we have $|g'||g_i'| \le (d_i+1)^{1/2}2^{d_i} |f_i'|$ for each $1\le i\le m$, where $d_i=\deg(f_i)$.
 Then $|g||g_i| = |{\bf LC}(g){\bf LC}(g_i) ||g'||g_i'| \le  (d_i+1)^{1/2}2^{d_i} | {\bf LC}(g){\bf LC}(g_i)||f_i'| =  (d_i+1)^{1/2}2^{d_i} |f_i|$.
  We have
  \begin{align}\label{heightamplify}
  \height(g)&\le \height(g)+\height(g_i)\notag\\ &\le \frac{1}{2}\log(d_i+1)+d_i\log 2+\height(f_i)~~\emph{{\rm for any $i$}}\\ &\le \frac{1}{2}\log(d+1)+d\log 2+h.\notag
 \end{align}
\end{proof}
\begin{remark}\label{factorheight}
 By equation \cref{heightamplify}, we have $\height(f_i/g)\le \frac{1}{2}\log(d+1)+d\log 2+h$ for any $i$.
\end{remark}

We now give the degree and height bounds for the \gHNF in $\Zx$.
\begin{lemma}  \label{heightgeneral}
Let $f_1,\ldots,f_m\in \Z[x]$ and $[g_1,\ldots,g_s]$ the {\rm GHNF} of $[f_1,\ldots,f_m]$.
Then  $\deg(g_i)\le d$ and $\height(g_i)\le (2d+1)(h+d\log 2+\log(d+1))$.
\end{lemma}
\begin{proof}
Obviously, the degree bound of the \gHNF in $\Z[x]$ is $d$ by the procedure of the \gb\, computation.
  Let $g=\gcd(f_1,\ldots,f_m)$ in $\Z[x]$,  then, $[g_1/g,\ldots,g_s/g]$ is the \gHNF of $[f_1/g,\ldots,f_m/g]$.
  By  Lemmas \ref{gcdheightbound} and \ref{factorheight},
   $\height(g)$ and $\height(f_i/g)$  are both $\le \frac{1}{2}\log(d+1)+d\log 2+h$.
  Moreover, $1\in (f_1/g,\ldots,f_m/g)\Q[x]$.
  By Lemma \ref{heightboundQx}, $\height(g_i/g)\le d(2(\frac{1}{2}\log(d+1)+d\log 2+h)+\log (d+1))= 2d(h+d\log 2+\log(d+1))$.
  So, $\height(g_i)\le 2d(h+d\log 2+\log(d+1))+\frac{1}{2}\log(d+1)+d\log 2+h\le (2d+1)(h+d\log 2+\log(d+1))$.
\end{proof}

Finally, we consider an effective Nullstellensatz in $\Zpx$,
whose proof follows that of Lemma 6.4 in \cite{aschenbrenner2004ideal}.
\begin{lemma}\label{degreeboundzpx}
If $1\in (f_1,\ldots,f_m)_{\Zpx}$, then there exist $h_1,\ldots,h_n\in \Zpx$ of degree at most $3d^2(2h+\log(d+1))/\log p$ such that
$1=f_1h_1+\cdots+f_mh_m$.
\end{lemma}
\begin{proof}
Suppose $1\in (f_1,\ldots,f_m)_{\Zpx}$, then   $1\in (f_1,\ldots,f_m)_{\Qx}$. By Lemma \ref{heightboundQx},
there exist $\delta\in \Z\setminus \{0\}$ with height $\le d(2h+\log (d+1))$
 and $g_1, \ldots,g_m\in \Zx$ with degrees $<d$ satisfying
 \begin{equation}\label{eq-delta}
   \delta=f_1g_1+\cdots+f_mg_m.
 \end{equation}
If $\delta$ is a unit in $\Zp$, then $$1=f_1(g_1/\delta)+\cdots+f_m(g_m/\delta).$$
Let $h_i=g_i/\delta$ for $i=1,\ldots,m$. Then we have the required properties. Suppose that $\delta$ is not a unit. Let $\mu=v_p(\delta)\ge 1.$
Clearly we have $1\in (f_1,\ldots,f_m)(\Zp/p\Zp)[x]$. Then by the Extended Euclidean Algorithm, there exist $r_1,\ldots,r_m\in \Z[x]$ with
$$1-(r_1f_1+\cdots+r_mf_m)\in (p)\Zpx$$ and $\deg(r_j)< d$ for all $j=1,\ldots,m.$
So there exists $s_1,\ldots,s_m\in \Zpx$ and $s\in (p^\mu)\Zpx$ such that
 \begin{equation}\label{eq-s}
  1-(f_1s_1+\cdots+f_ms_m)=s.
\end{equation}
We have $\deg(s_j)\le \mu(2d-1)-d$ for all $j$; hence $\deg(s)\le \mu(2d-1)$. By equations \cref{eq-delta} and \cref{eq-s}, we have
$$1=f_1s_1+\cdots+f_ms_m+s=f_1h_1+\cdots+f_mh_m$$ with $h_j=s_j+(s/\delta)g_j\in \Zpx$. We have
$$\deg(sg_j)\le \mu(2d-1)+d\le 3\mu d.$$
Since $\mu \log p\le \height(\delta)\le d(2h+\log (d+1))$, it follows that $\deg(h_j)$ is bounded by $3d^2(2h+\log (d+1))/\log p$.
\end{proof}

Then we can give the degree bound for the global case.
  \begin{lemma}
  \label{degreeboundforU}
  If $1\in (f_1,\ldots,f_m)_{\Z[x]}$, then there exist $h_1,\ldots,h_m\in \Z[x]$ such that $1=f_1h_1+\cdots+f_mh_m$, with
  $\deg(h_i)\le 3d^2(2h+\log (d+1))$ for $i=1,\ldots,m$.
\end{lemma}
\begin{proof}
  By Lemma \ref{heightboundQx}, we have $g_1,\ldots,g_m\in \Zx$ with degrees $<d$ and $\delta\in \Z$ satisfying
  $$\delta=f_1g_1+\cdots+f_mg_m.$$
  Let $p_1,\ldots,p_k$ be all the prime factors of $\delta$.
  Since $1\in (f_1,\ldots,f_m)_{\Zx}$,
  we have $1\in (f_1,\ldots,f_m)_{\Z_{(p_i)}[x]}$.
  By Lemma \ref{degreeboundzpx}, there exist $h_1^{(p_i)},\ldots,h_m^{(p_i)}\in \Zx$ with degrees $\le 3d^2(2h+\log (d+1))/\log p_i$
  and $\delta^{(p_i)}\in \Z\setminus(p)\Z$ satisfying $\delta^{(p_i)}=f_1h_1^{(p_i)}+\cdots+f_mh_m^{(p_i)}$.
  Then there exist $a,a_1,\ldots,a_k\in \Z$ satisfying $$1=a\delta+a_1\delta^{(p_1)}+\cdots+a_k\delta^{(p_k)}.$$
  Hence letting $h_j=ag_j+a_1h_j^{(p_1)}+\cdots+a_kh_j^{(p_k)}\in \Zx$ for $j=1,\ldots,m$, we get
  $1=f_1h_1+\cdots+f_mh_m.$
  From this, we can easily get $\deg(h_i)\le 3d^2(2h+\log (d+1))$ for $i=1,\ldots,m$.
\end{proof}

\subsection{Degree and height bounds for solutions to linear equations over $\Zx$}
Throughout this section, let $F=(f_{ij})\in \Zx^{n\times m}$,
$d=\deg(F)$  the maximal degree of elements in $F$,
and $h=\height(F)$ the maximal height of elements in $F$.
For anysubring $R$ of $\C$, let
$$\sol_{R[x]}(F) =\{Y\in R[x]^m\,|\, FY=0\}$$
which is an $R[x]$-module in $\Zx^m$.
Let $r$ be the rank of $F$ and $F_1$ the matrix consisting of $r$ linear
independent rows of $F$. Then, $\sol_{R[x]}(F) =\sol_{R[x]}(F_1)$.
So, we may assume $r=n$ unless mentioned otherwise.
In this section, we will show that $\sol_{R[x]}(F)$
has a set of generators whose degrees and heights can be nicely bounded.

For a prime $p$, $f=\sum_{v=0}^{\infty}f_vx^v\in \hatZpanglex$ is called \emph{regular of degree $s$ with respect to $p$}, or simply,
\emph{regular of degree $s$} when there is no confusion,
if its reduction $\overline{f}\in \hatZpanglex/p \hatZpanglex$ is unit-monic of degree $s$, that is,
$\overline{f_s}\neq 0$, and
$v_p(f_i)>0$ for all $i>s$, where $v_p$ is the $p$-valuation.
Now we describe the Weierstrass Division Theorem for $\hatZpanglex$:
\begin{theorem}[\cite{aschenbrenner2004ideal,bosch1984non}]\label{Weiestrasstheorem}
Let $g\in \hatZpanglex$ be regular of degree $s$. Then for each $f\in \hatZpanglex$, there are uniquely determined elements $q\in \hatZpanglex$ and $r\in \hatZp[x]$ with $\deg(r)<s$ such that $f=qg+r$.
\end{theorem}

 \begin{lemma}\label{homogeneousZpxcompletion}
$\sol_{\hatZpanglex}(F)$ has a set of generators in $\Zx^m$ with degrees $\leq nd$.
\end{lemma}
\begin{proof}
 Let $\triangle$ be an $n\times n$-submatrix of $F$ with $\delta=\det(\triangle)\neq 0$ having the least $p$-valuation among all the nonzero $n\times n$ minors of $F$.
 After permutating the unknowns of $y_1,\cdots,y_m$ in $Fy=0$, we may assume $\triangle =(f_{ij})_{1\leq i,j\leq n}$.
 Multiplying both sides of $Fy=0$ on the left by the adjoint of $\triangle$, the system $Fy=0$ becomes
\begin{equation}\label{By=0Zphatx}
\left(
  \begin{array}{cccccc}
    \delta & & & c_{1,n+1} & \cdots & c_{1,m} \\
     & \ddots& & \vdots &  & \vdots \\
    & & \delta & c_{n,n+1} & \cdots & c_{n,m} \\
  \end{array}
\right)
\left(
  \begin{array}{c}
    y_1 \\
    \vdots \\
    y_m \\
  \end{array}
\right)=\left(
  \begin{array}{c}
    0 \\
    \vdots \\
    0 \\
  \end{array}
\right)
\end{equation} where $\delta$ and all the $c_{ij}$ are in $\Zx$ with degrees $\le nd$. Note that, $v_p(c_{ij}) \ge v_p(\delta)$ for all $i, j$, by the choice of $\triangle$.
Let
\begin{equation}\label{equation3}
v^{(1)}=\left(
          \begin{array}{c}
            -c_{1,n+1} \\
            \vdots \\
            -c_{n,n+1} \\
            \delta \\
            0 \\
            \vdots \\
            0 \\
          \end{array}
        \right),\ldots,v^{(m-n)}=\left(
          \begin{array}{c}
            -c_{1,m} \\
            \vdots \\
            -c_{n,m} \\
            0 \\
            \vdots \\
            0 \\
            \delta \\
          \end{array}
        \right).
\end{equation} Then, $Fv^{(i)}=0$ for $i=1,\ldots,m-n$ and $v^{(1)},\ldots,v^{(m-n)}$ are in the $\hatZpanglex$-module $\sol_{\hatZpanglex}(F)$. Let $\mu =v_{p}(\delta), u^{(i)}=p^{-\mu}v^{(i)}$ for $i=1,\ldots, m-n$.
Then $u^{(1)},\ldots,u^{(m-n)}$ are also in $\sol_{\hatZpanglex}(F)$.
Multiplying the equation~\cref{By=0Zphatx} by $p^{-\mu}$, we have $By=0$, where $B=
\left(\begin{array}{cccccc}
   \varepsilon & & & d_{1,n+1} & \cdots &d_{1,m} \\
     & \ddots& & \vdots &  & \vdots \\
    & & \varepsilon & d_{n,n+1} & \cdots & d_{n,m} \\
  \end{array}\right)$ and $\varepsilon$ is regular of degree $s$ for some integer $s\leq nd$.
   Clearly, the $(n+i)$-th element of $u^{(i)}$ is $\varepsilon$.
   Moreover, $\varepsilon$ and all the $d_{ij}$ are in $\Zx$ with degrees $\le nd$

In the system $Fy=0$, let
$$f_{ij}=f_{ij0}+\cdots+f_{ijd}x^d,\quad y_j=y_{j0}+\cdots+y_{j,nd-1}x^{nd-1}$$
 for $1\leq i\leq n,~1\leq j\leq m$, where $f_{ijk}\in \Zp$ and $y_{jk}$ are the new unknowns in $\hatZpanglex$.
 The $i$-th equation in $Fy=0$ may then be written as
 $$ \sum_{l=0}^k\sum_{j=1}^m f_{ijl}y_{j,k-l} = 0 ,  \hskip 2.0cm 0\le k < (n+1)d, $$
 where we put $f_{ijl}=0$ for $l > d$ and $y_{jl}=0$ for $l\ge nd$.
  Then we obtain a new system $F'y'=0$, where $F'\in \Zp^{(nd(n+1))\times (mnd)}$, $y'=[y_{10},\ldots,y_{1,nd-1},\ldots,y_{m0},\ldots,y_{m,nd-1}]^\tau$,
  whose solutions in $\hatZp$ are in a one to one correspondence with the solutions of $Fy=0$ in $\hatZp[x]$ of degrees $<nd$.
  We have a set of finite generators for $F'y'=0$, thus we have finitely many solutions $y^{(1)},\ldots,y^{(M')}\in \Zpx^m$ of $Fy=0$ such that each solution to $Fy=0$ of degree $< nd$ is a $\hatZp$ linear combination of $y^{(1)},\ldots,y^{(M')}$.

We claim that the above
$ u^{(1)},\ldots, u^{(m-n)},y^{(1)},\ldots,y^{(M')}$
  generate the $\hatZpanglex$-module $\sol_{\hatZpanglex}(F)$. So $\sol_{\hatZpanglex}(F)$ can be generated by elements in $\Zpx^m$ of degrees $\leq nd$.

Now we prove the claim.
  Let $w=[w_1,\ldots,w_m]^\tau\in \hatZpanglex^m$ be any solution to $Fy=0$. Since $\varepsilon$ is regular of degree $s$ for some integer $s\leq nd$,
  by Theorem \ref{Weiestrasstheorem}, there exist $Q_{n+1},\ldots, Q_m\in \hatZpanglex$ and $R_{n+1},\ldots,R_m\in \hatZp[x]$ whose degrees are less than $ s$
  such that $R_j=w_j-Q_j\varepsilon$ for $j=n+1,\ldots,m$.
  Let $z=w-Q_{n+1}u^{(1)}-\cdots-Q_mu^{(m-n)}=[h_1,\ldots,h_n,R_{n+1},\ldots,R_m]$, which is obvious a solution to $By=0$.
  So we have $\varepsilon h_i=-d_{i,n+1}R_{n+1}-\cdots-d_{i,m}R_{m}$ for $i=1,\ldots,n$.
  Since $\varepsilon, d_{ij}$ are in $\hatZp[x]$ with degrees $\le nd$ and $R_j\in \hatZp[x]$ are of degrees $< s$,
  we have $\deg(h_i)<nd$ for $i=1,\ldots,n$.
  Hence $\deg(z)< nd$, therefore it can be expressed as the $\hatZp[x]$ combination of $y^{(1)},\ldots,y^{(M')}$.
  Now it is clear that $w$ is the $\hatZp[x]$ combination of $u^{(1)},\ldots, u^{(m-n)},y^{(1)},,\ldots,y^{(M')}$.
  Hence $\sol_{\hatZpanglex}(F)$ as a $\hatZpanglex$-module can be generated by $u^{(1)},\ldots, u^{(m-n)},$ $y^{(1)},,\ldots,y^{(M')}$.
\end{proof}

In the proof of Lemma \ref{homogeneousZpxcompletion}, if we choose $\triangle$ to be any $n\times n$-submatrix of $F$ whose determinant is nonzero, let $\mu=0$ and do the computations in $\Qx$, we can easily give the following lemma:
\begin{lemma}
\label{homogeneousQx}
$\sol_{\Qx}(F)$ can be generated by elements in $\Zx^m$ of degrees $\le nd$.
\end{lemma}
Now we describe Corollary 2.7 of \cite{aschenbrenner2004ideal} in our notations:
\begin{lemma}[\cite{aschenbrenner2004ideal}]\label{hattononhat}
  Let $F$ be an $n\times m$ matrix over $\Zpx$. If $y^{(1)},\ldots,y^{(L)}\in \Zpx^m$ generate
   the $\Q[x]$-module $\sol_{\Qx}(F)$
    and $z^{(1)},\ldots,z^{(M)}\in \Zpx^m$ generate the $\hatZpanglex$-module $\sol_{\hatZpanglex}(F)$.
    Then $y^{(1)},\ldots,y^{(L)},$ $z^{(1)},\ldots,$ $z^{(M)}$ generate the $\Zpx$-module $\sol_{\Zpx}(F)$.
\end{lemma}
By Lemmas \ref{homogeneousZpxcompletion}, \ref{homogeneousQx}, and \ref{hattononhat}, we have the following corollary:
\begin{coro}\label{Zpxhomogeneous}
  $\sol_{\Zpx}(F)$ can be generated by elements in $\Zx^m$ of degrees $\le nd$.
\end{coro}

We describe Lemma 4.2 of \cite{aschenbrenner2004ideal} in our notations as follows:
\begin{lemma}\label{Lemma4.2of[2]}
Let $M$ be a $\Zx$-submodule of $\Zx^m$. For each maximal ideal $(p)$ of $\Z$, let $u_p^{(1)},\ldots,u_p^{(K_p)}\in M$ generate the $\Zpx$-submodule $(M)_{\Zpx}$ of $\Zpx^m$.
   Then $u_p^{(1)},\ldots,u_p^{(K_p)}$, where $(p)$ ranges over all maximal ideals of $\Z$, generate the $\Zx$-module $M$.
\end{lemma}

We now give a degree bound for the solutions of linear equations over $\Zx$.
\begin{coro}\label{Zxcase}
Let $F=(f_{ij})\in \Zx^{n\times m}$ and $d=\deg(F)$.
Then $\sol_{\Zx}(F)$ can be generated by a finite set of elements whose degrees are $\le nd$.
\end{coro}
\begin{proof}
  By Lemmas \ref{Zpxhomogeneous} and \ref{Lemma4.2of[2]}, we  know that $\sol_{\Zx}(F)$ can be generated by elements whose degrees are $\le nd$.
  Since $\sol_{\Zx}(F)\subset \Zx^m$ and $\Zx^m$ is Noetherian, the set of generators must be finite.
\end{proof}

\begin{remark}
In results \ref{homogeneousZpxcompletion}, \ref{homogeneousQx}, and \ref{Zpxhomogeneous}, \ref{Zxcase}, if $F$ is of rank $r$, then the generators can be bounded by $rd$.
\end{remark}

In the rest of this section, we give height bounds for $\sol_{\Z[x]}(F)$.
By Remarks of Corollary 1.5 and Lemma 5.1 in \cite{aschenbrenner2004ideal},
we have the following result.
\begin{lemma}[\cite{aschenbrenner2004ideal}]\label{heightboundZ}
Let $A\in \Z^{n\times m}$, $r=\rank(A)$, and $h=\height(A)$. Then $\sol_{\Z}(A)$ can be generated by $m-r$  vectors whose heights are bounded by $2r(h+\log r+1)$.
\end{lemma}

Let $F=(f_{ij})\in \Z[x]^{n\times m}$, $d=\deg(F)$, $h=\height(F)$, and $F$ is of full rank.
Then, we have
\begin{theorem}\label{homogeneousZx}
$\sol_{\Z[x]}(F)$ can be generated by vectors whose degrees are bounded by $nd$ and heights are bounded by $2(n(n+1)d+n)(h+\log (n(n+1)d+n)+1)$.
\end{theorem}
\begin{proof}
By  Corollary \ref{Zxcase}, $\sol_{\Z[x]}(F)$ can be generated by elements of degrees $\le nd$.
Let $[y_1,\ldots,y_m]^\tau\in \sol_{\Zx}(F)$.
Assume $f_{ij}=a_{ij0}+a_{ij1}x+\cdots+a_{ijd}x^d$, $y_j=y_{j0}+y_{j1}x+\cdots+y_{j,nd}x^{nd}$, where $a_{ijk}\in \Z$, $y_{jk}$ are indeterminants taking values in $\Z$.
Then, $Fy=0$ can be written as the following matrix equation
\begin{equation}\label{equationZy'}
\left(
\begin{array}{c}
 A_1 \\
 \vdots \\
  A_n \\
  \end{array}
  \right)y'=0,
  \end{equation}  $y'=[y_{10},\ldots,y_{1,nd},\ldots,y_{m0},\ldots,y_{m,nd}]^\tau$, $A_i=[A_{i1},\ldots,A_{im}]_{((n+1)d+1)\times (m(nd+1))}$, and
$$
A_{ij}=\left(
           \begin{array}{cccc}
             a_{ij0} &  &  &  \\
             a_{ij1} & a_{ij0} &  &  \\
             \vdots &  & \ddots &  \\
             a_{ijd} &  &  & a_{ij0} \\
              & \ddots &  & \vdots \\
              & &  & a_{ijd} \\
           \end{array}
         \right)_{((n+1)d+1)\times (nd+1)}
$$
for $i=1,\ldots,n$. So
$\left(
 \begin{array}{c}
   A_1 \\
     \vdots \\
       A_n \\
       \end{array}
       \right)\in \Z^{(n(n+1)d+n)\times (m(nd+1))}$.
By Lemma \ref{heightboundZ}, we have the equation system \cref{equationZy'} can be generated by vectors whose heights are bounded by $2(n(n+1)d+n)(h+\log (n(n+1)d+n)+1)$.
\end{proof}

\begin{remark}
Let $D=\Z[x_1,\ldots,x_N]$ and $A\in D^{n\times m}$.
In \cite{aschenbrenner2004ideal},  Aschenbrenner proved that
$\sol_{D}(A)$ has a set of  generators whose
degrees and heights are bounded by $(2nd)^{2((N+1)^N-1)}$ and $C_2(2n(d+1))^{(N+1)^{O(N)}}(h+1)$,
respectively,
where $C_2$ is a constant only depending on $A$, $d=\deg(A)$, $h=\height(A)$.
Setting $N=1$ in these bounds, we obtain the degree and height bounds $(2nd)^2$ and $C_2(2n(d+1))^{2^{O(1)}}(h+1)$, respectively. Due to the special structure
of the Gr\"obner basis in $\Zx$, our results are much better than that of \cite{aschenbrenner2004ideal} in the $\Zx$ case.
\end{remark}

Let $F\in \Z[x]^{n\times m}$, $b\in \Z[x]^m$. Denote  $d=\deg(F,b) = \max(\deg(F), \deg(b)),~h=\max(\height(F), $ $ \height(b)).$
Similar to Theorem 6.5 in \cite{aschenbrenner2004ideal}, we have the following degree bound.
\begin{theorem}\label{degreeboundZx}
If the system $Fy=b$ has a solution in $\Z[x]^m$,
then it has such a solution of degree $\le 3n^2d^2(h_2+\log (nd+1))+nd$,
where $h_2=2(n(n+1)d+n)(h+\log(n(n+1)d+n)+1)$.
\end{theorem}
\begin{proof}
By Theorem \ref{homogeneousZx}, there exist generators $z^{(1)},\ldots,z^{(K)}$ for the $\Zx$-module of solutions to the system of $(F,-b)z=0$, where $z^{(k)}=[z_1^{(k)},\ldots,z_{m+1}^{(k)}]^\tau$ is a vector of $m+1$ unknowns, with $\deg(z^{(k)})\le nd$  and
  \begin{align}
    \height(z^{(k)})&\le 2(n(n+1)d+n)(h+\log (n(n+1)d+n)+1)=h_2\notag
  \end{align}
  for all $k=1,\ldots,K$.
  For each $k$, let $z_{m+1}^{(k)}\in \Z[x]$ be the last component of $z^{(k)}$.
  Clearly, $Fy=b$ is solvable in $\Zx$ if and only if $1\in (z_{m+1}^{(1)},\ldots,z_{m+1}^{(K)})$.
  Moreover, if $h_1,\ldots,h_K$ are elements of $\Zx$ such that $1=h_1z_{m+1}^{(1)}+\cdots+h_Kz_{m+1}^{(K)}$, then $[y,1]^\tau=h_1z^{(1)}+\cdots+h_Kz^{(K)}$ is a solution to $Fy=b$.
  By Lemma \ref{degreeboundforU}, we have $$\deg(h_k)\le 3n^2d^2(2h_2+\log (nd+1)),$$ where $h_2=2(n(n+1)d+n)(h+\log (n(n+1)d+n)+1)$.
  It follows that $\deg(y)\le 3n^2d^2(2h_2+\log (nd+1))+nd$.
\end{proof}

\subsection{Degree and height bounds in $\Zx^n$}
In this section, we assume $F=(f_{ij})\in \Z[x]^{n\times m}$, $d=\deg(F)$,
$h=\height(F)$, and $F$ is of full rank.
Let  $\mathcal{C}$ in~\cref{reducedGB} be the \gHNF of $F$.
We will give degree and height bounds for $\mathcal{C}$.

In our analysis of the complexity, only the degree and height bounds of
$c_{r_i,k_i}$ in the $r_i$-th rows of $\mathcal{C}$ will be used.
So, we define $\deg(\mathcal{C}) = \max_{i,k_i} \deg(c_{r_i,k_i})$
 and
$\height(\mathcal{C}) = \max_{i,k_i} \height(c_{r_i,l_i})$.
The following theorem gives the degree and height bounds for the \gHNF of $F$.
\begin{theorem}\label{degreeheightboundforGHNF}
We have $\deg(c_{r_i,l_i})\leq (n-r_i+1)d$ and $\height(c_{r_i,j})\le 6(n-r_i+1)^3d^2(h+1+\log((n-r_i+1)^2 d))$ for any $1\le i\le t,~1\le j\le l_i$.
\end{theorem}
\begin{proof}
Without loss of generality, we need only to prove the theorem for $r_1=1$,
in which case $\deg(c_{1j})\le nd$ and $\height(c_{1j})\le 6n^3d^2(h+1+\log(n^2 d))$ for $1\le j\le l_1$.

For any $[a,0,\cdots,0]^\tau\in (F)$, which is the $\Zx$ lattice generated by the columns of $F$,
there exists a $\u\in \Z[x]^m$, such that $[a,0,\cdots,0]^\tau=Fu$ and hence $\u\in \sol_{\Zx}(F_{n-1})$,
where $F_{n-1}$ is the last $n-1$ rows of $F$.
By Theorem \ref{homogeneousZx},
$\sol_{\Zx}(F_{n-1})$ can be generated by polynomials of degrees $\leq (n-1)d$ and heights $\le h_1=2(n(n-1)d+(n-1))(h+\log (n(n-1)d+(n-1))+1)$,
say $\{v^{(1)},\ldots,v^{(s)}\}$.
Then, $[a,0,\ldots,0]^\tau\in (F)$ can be generated by $\{Fv^{(1)},\ldots,Fv^{(s)}\}$ and $\deg(Fv^{(j)})\le nd$ and $\height(Fv^{(j)})\le h+h_1$.
Let $\deg(Fv^{(j)})=[t_j,0,\ldots,0]^\tau$ for some $t_j\in \Zx$, $1\le j\le s$.
Then, $[c_{1,1},\ldots,c_{1,l_1}]$ is the \gHNF of $[t_1,\ldots,t_s]$, and $\deg(t_j)\le nd$, $\height(t_j)\le h+h_1$.
By Lemma \ref{heightgeneral}, we have $\deg(c_{1,j})\leq nd$, $i.e.~\deg(c_{1,j})\leq nd$ for $1\le j\le l_1$.
Moreover, \begin{align*}
&\height(c_{1j}) \\
&\le (2nd+1)(h+h_1+nd\log 2+\log(nd+1))\\
&=(2nd+1)(h+2(n(n-1)d+(n-1))(h+\log (n(n-1)d+(n-1))+1)\\
&~~~~+nd\log 2+\log(nd+1))\\
&\le (2nd+1)(h+2n^2d(h+\log(n^2d)+1)+nd\log 2+\log(n^2d))\\
&\le 6n^3d^2(h+1+\log(n^2 d)).
\end{align*}
\end{proof}

\begin{remark}\label{betterheightboundforGHNF}
Note that, since the last $n-r_i+1$ rows of $F$ have rank $t-i+1$, by the above proof, we have
$\deg(c_{r_i,j})\le (t-i+1)d$ and
$\height(c_{r_i,j})\le 6(t-i+1)^3d^2(h+1+\log((t-i+1)^2 d))$ where $h=\height(F)$, for $1\le i\le t$, $1\le j\le l_i$.
\end{remark}

We have the following degree bound for the transformation matrix $U$, which satisfying $\mathcal{C}=FU$.
\begin{theorem}\label{degreeboundtranmat}
  Let $F\in \Z[x]^{n\times m}$, $\mathcal{C}$ its \gHNF, and
  $U\in \Z[x]^{m\times s}$ the transformation matrix satisfying $\mathcal{C}=FU$.
   Then, $\deg(U) \le D$, where $D=73n^8d^5(h+1+\log(n^2d))$.
    \end{theorem}
\begin{proof}
  By Theorem \ref{degreeheightboundforGHNF}, we have
  $\deg(c_{r_i,j})\le (n-r_i+1)d$, $\height(c_{r_i,j})\le 6(n-r_i+1)^3d^2(h+1+\log((n-r_i+1)^2 d))$ for $i=1,\ldots,t,~j=1,\ldots,l_i$.
  Denote by $U_{r_i,j}$ the column vector of $U$, satisfying $FU_{r_i,j}=[\ast,\ldots,\ast,$ $c_{r_i,j},0,\ldots,0]^\tau$.
   Then $U_{r_i,j}$ can be determined by $F_{n-r_i+1}U_{r_i,j}=[c_{r_i,j},0,\ldots,0]^\tau$, where $F_{n-r_i+1}$ is the last $n-r_i+1$ rows of $F$.
   In Theorem \ref{degreeboundZx}, let $\deg(F,b)=\max_{i,j}\deg(F,c_{r_i,j})\le nd,
   ~\height(F,b)=\max_{i,j}\height(F,c_{r_i,j})\le 6n^3d^2(h+1+\log(n^2 d))$.
   Then we have $\deg(U)\le 3n^2d^2(h_2+\log(nd+1))+nd$, where $h_2=2(n(n+1)$ $\deg(F,b)+n)(\height(F,b)+\log(n(n+1)\deg(F,b)+n)+1)$.
   First, we have the following inequality:
   \begin{align}
   h_2&= 2(n(n+1)\deg(F,b)+n)(\height(F,b)+\log(n(n+1)\deg(F,b)+n)+1)\notag\\
   &\le 2(n^2(n+1)d+n)(6n^3d^2(h+1+\log(n^2 d))+\log(n^2(n+1)d+n)+1)\notag\\
   &\le 24n^6d^3(h+1+\log n^2d) \hskip 3.0cm  ~\emph{\rm{for any}}~n\ge 2.
   \end{align}
   One can verify that the above inequality is still valid for $n=1$, in which case $\deg(F,b)\le d$ and $\height(F,b)\le d(2h+\log (d+1))+\frac{1}{2}\log (d+1)+d\log d+h$.
    So we have
   $\deg(U)\le 3n^2d^2h_2+3n^2d^2\log (nd+1)+nd \le73n^8d^5(h+1+\log n^2d)$.
\end{proof}

We give an example to illustrate the main idea of the proof.
\begin{example}
Let $F=\left(
         \begin{array}{cc}
           1 & x \\
           6x^3+1 & 8x^2 \\
         \end{array}
       \right)$,
and $h=3\log 2 =3$ the height of $F$, where we choose the logarithm with $2$ as a base.

If $a=[a_1,a_2]^\tau$ with $a_2\ne 0$ is a column vector of $\mathcal{C}$, then $a_2$ is an element of the \gHNF of $[6x^3+1,8x^2]$.
Thus, $\deg(a_2)\le \max(\deg(6x^3+1),\deg(8x^2))=3$ and by Theorem \ref{gcdheightbound}, $\height(a_2)\le 4\log 2+h = 7$.

If $b=[b_1,0]^\tau$ with $b_1\ne 0$ is a column of $\mathcal{C}$,
then  there exists a $U=[u_1,u_2]^\tau\in \Z[x]^2$ satisfying
\begin{center}
$b=FU$,  $~~~~i.e.~~~~$
$\left\{
    \begin{aligned}
    b_1&=u_1+x u_2\\
    0&=(6x^3+1)u_1+8x^2u_2
    \end{aligned}
\right.$
\end{center}
Let ${\bf g}_1,\ldots,{\bf g}_s$ be the generators of the solutions to $0=(6x^3+1)u_1+8x^2u_2$.
By Theorem \ref{homogeneousZx}, $\deg( {\bf g}_i)\le 3$ and  $\height( {\bf g}_i)\le 14(h+\log 7+1)$.
 Thus, $b_1$ is an element of the \gHNF of $[1,x]\cdot [{\bf g}_1,\ldots,{\bf g}_s]=[h_1,\ldots,h_s]$, where $\deg(h_i)\le 4$, and $\height(h_i)\le 28(h+\log 7+1) < 196$.
 Hence, by Theorem \ref{degreeheightboundforGHNF}, $\deg(b_1)\le 4$ and $\height(b_1)\le 432(h+1+\log 12)<3456$.
 Moreover, by Theorem \ref{degreeboundtranmat}, we know that the degree bound for the transformation matrix is $D=4478976(h+1+\log 12)< 35831808$.

Actually, the solutions to $0=(6x^3+1)u_1+8x^2u_2$ can be generated by $[8x^2,-(6x^3+1)]^\tau$. Thus, $b_1$ is in the \gHNF of $[1,x]\cdot[8x^2,-(6x^3+1)]^\tau=[-6x^4+8x^2-x]$.
 The \gHNF and the  transformation matrix are
 \begin{small}
 $$\mathcal{C}=\left(
         \begin{matrix}
           6x^4-8x^2+x & 3x^8-4x^6+5x^5-6x^3+1 \\
           0 & 1 \\
         \end{matrix}
       \right),U=\left(
                                            \begin{matrix}
                                              -8x^2 & -4x^6-6x^3+1 \\
                                              6x^3+1 & 3x^7+5x^4 \\
                                            \end{matrix}
                                          \right).$$
 \end{small}
\end{example}
So for some examples, the bounds are far from optimal, and this is
the reason we will give an incremental algorithm in the next section to compute the GHNF.

\section{Algorithms to compute the GHNF}
\label{HNF Algorithm in matrix form}
In this section, we give an  algorithm to compute the GHNF of $F\in \Zx^{n\times m}$.
Roughly speaking, the algorithm works as follows.
We will compute the HNF  $G\in\Z^{s\times k}$ for the coefficient matrix of $F$
and check whether a GHNF can be retrieved from $G$.
In the negative case, certain prolongations are done to $G$ and the procedure is repeated.
The key idea is how to do the prolongation so that the sizes of
the matrices $G$ are nicely controlled.

\subsection{HNF of integer matrix}
In this section, we will introduce several basic results about
HNF of an integer matrix, which will be used as the main computational tool
in our GHNF algorithm.

\begin{myDef}\label{def-hnf}
A matrix $H=(h_{i,j})\in\Z^{n\times m}$ is called an (column) HNF if there
exists an $r \le m$ and a strictly increasing map $f$ from $[r+1,m]$ to $[1,n]$
satisfying:
(1) for $j\in[r+1,m]$, $h_{f(j),j}\ge 1$, $h_{i,j}=0$ if $i > f(j)$
  and $h_{f(j),j}> h_{f(j),k}\ge0$ if $k > j$; and (2)
the first $r$ columns of $H$ are equal to zero.
%
%
\end{myDef}

Let $A\in\Z^{n\times m}$ and $H^{n\times m}$ the HNF of $A$.
Then there exists a $U\in {\rm GL}_m(\Z)$  \cite{cohen1993course} such that
\begin{equation}\label{eq-hnf}
H=AU.
\end{equation}
Note that $H$ is obtained from $A$ by doing column elementary operations
which are represented by the matrix $U$.
We need the following lemma on the syzygy module of $A$.
\begin{lemma}\cite{cohen1993course}
\label{HNF over integer}
Let \cref{eq-hnf} be given and assume that the first $r$ columns of $H$ are the $\0$ columns of $H$.
Then a $\Z$-basis for the $\Z$-module $\Syz(A)=\{Y\in\Z^m\,|\, AY=0\}$ is given by the first $r$ columns of $U$.
\end{lemma}

We will measure the cost of our algorithms in numbers of bit operations.
We need the function $M(k)=O(k\log k\log \log k)$ which is the cost of multiplications and quotients of two integers $a$ and $b$ with $|a|,|b|<2^k$.
We will give complexity results in terms of the function $B(k)=M(k)\log k=O(k(\log k)^2(\log \log k))$.
We use a parameter $\theta$ such that the multiplication of two $n\times n$ integer matrices needs $O(n^{\theta})$ arithmetic operations.
The  best known upper bound for $\theta$ is about 2.376.

The following result gives the complexity of computing HNF over $\Z$.

\begin{theorem}[\cite{Storjohann2013}]\label{th-hnfc}
Let $A\in\Z^{n\times m}$ with rank $r$ and height $h$, and $H$ be the HNF of $A$.
Then $\height(H)\le \log \beta= r(\frac{1}{2}\log r+h)$.
The bit complexity to compute $H$ from $A$ is
$O(mnr^{\theta-2} \log \beta M(\log \log \beta)/\log \log \beta +mnB(\log \beta)\log r)$.
\end{theorem}

\subsection{The $\Zx$  case}
In this section, we will show how to compute the GHNF in $\Zx$.
Through out this section, let $F=[f_1,\ldots,f_m]$ be a polynomial vector over $\Zx$,
$d=\deg(F)$, and $h=\height(F)$.
$C\in \Z^{(d+1)\times m}$ is called the \emph{coefficient matrix} of $F$
if its columns represent the polynomials in $F$ such that
$$F=\X_d C, \hbox{ where }\X_d=[1, x, \ldots, x^d].$$
Let $[\0,H]\in\Z^{(d+1)\times m}$ be the HNF of $C$, where  $H\in\Z^{(d+1)\times s}$
contains no zero columns.
Then, there is a unimodular matrix $U=[U_{1},U_{2}]$ such that
$[\0,H]=C U$, $\0=C U_1$, and $H=CU_2$.
We call $G=\X_d H$  the \emph{polynomial Hermite normal form} (abbr. PHNF) of $F$.
For simplicity, we denote   $C=\CM(F)$ and
\begin{equation}\label{eq-ph1}
G=\hf(F)=\X_dH=\X_dCU_2=FU_2.
\end{equation}
Let $G=[g_1,\ldots,g_s]\in\Zx^{1\times s}$. From the definition of HNF,
we have $\deg(g_1) < \deg(g_2) < \cdots < \deg(g_s)$.
We now give the algorithm.
\begin{algorithm}[!htb]
\label{HNF1new}
\caption{GHNF$_1(F)$}
\begin{algorithmic}[1]
\REQUIRE~~$F=[f_1,\ldots,f_m],~f_i\in \Zx$ and $d=\max_i \deg(f_i)$.
\ENSURE~~The \gHNF, or the reduced Gr\"obner basis, of $F$.

 \STATE  Let $G_0=\hf(F)$ and $k=0$.\\
 \STATE(loop) \label{iteration}
 $k=k+1$.\\
 $P_k=[G_{k-1},xG_{k-1,d-1}]$, where $G_{k-1,d-1}$ is the set of polynomials in $G_{k-1}$ with degrees $\le d-1$.\\
 $G_k=\hf(P_k)$.\\
 If $G_{k}\ne G_{k-1}$, repeat Step 2.
 \STATE\label{pick}
 Let $G_k=[g_{1},\ldots,g_{s}]$ and $R=[g_{1}]$.
 For $j$ from $2$ to $s$, if $\lc(g_{j-1})\ne\lc(g_{j})$,   $R = R \cup\{ g_{j}\}$.
 \STATE Return $R$.\\

\vskip5pt
(For $F=[f_1,\ldots,f_m]$ and $G=[g_1,\ldots,g_m]$, we use the notation
$[F,G]=[f_1,\ldots,f_m,g_1,\ldots,g_m]$.)
 \end{algorithmic}
 \end{algorithm}

\begin{example}\label{ex-11}
  $F=[6x^3+3x^2+12,6x^3+3x^2+6x,6x^3+15x^2,6x^3+3x^2]$.

  Step 1:
     $G_0=\text{\hf}(F)=[12,6x,12x^2,6x^3+3x^2]$. We have $d=3$.

  $1$-th loop:
    $P_1=[G_0,12x,6x^2,12x^3]$,
    $G_1=\text{\hf}(P_1)=[12,6x,6x^2,6x^3+3x^2]$.

  $2$-th loop:
    $P_2=[G_1,12x,6x^2,6x^3]$,
    $G_2=\text{\hf}(P_2)=[12,6x,3x^2,6x^3]$.

  $3$-th loop:
    $P_3=[G_2,12x,6x^2,3x^3]$,
    $G_3=\text{\hf}(P_3)=[12,6x,3x^2,3x^3]$.

  $4$-th loop:
    $P_4=[G_3,12x,6x^2,3x^3]$,
    $G_4=\text{\hf}(P_4)=[12,6x,3x^2,3x^3]$. The loop is terminated.

   Step 3: $R=[12,6x,3x^2]$ is the \gHNF of $F$.
\end{example}

In the rest of this section, we will prove the correctness of the algorithm
and give its complexity.

For a polynomial vector $F=[f_1,\ldots,f_m]$,
we denote $(F)_{\Z}$ to be $\Z$-module generated by
the elements of $F$.
If $\deg(f_i) < \deg(f_j)$ for all $i < j$, $F$ is called a
{\em $\Z$-Gr\"obner basis} for the following reason:
if $F$ is a $\Z$-Gr\"obner basis and $f\in(F)_{\Z}$,
then there exists an $f_k$ such that $\lt(f_k)|\lt(f)$,
or equivalently, $f$ can be reduced to zero by $F$ over $\Z$.
Furthermore, if $\lt(f_i)$ is not a $\Z$-factor of any monomial
of $f_j$ for $j\ne i$, then $F$ is called a {\em reduced $\Z$-Gr\"obner basis}.
By  Definition \ref{def-hnf} and \cref{eq-hnf}, we have
\begin{lemma}\label{lm-phnf}
Let $G=\hf(F)$. Then $(F)_{\Z}=(G)_{\Z}$ and $G$ is a reduced $\Z$-Gr\"obner basis of
$(F)_{\Z}$.
\end{lemma}

In Step 2 of Algorithm GHNF$_1$, if using the following ``full'' prolongation
in the $k$-th loop, we have
\begin{equation}\label{eq-lpf1}
  \widetilde{P}_k=[\widetilde{G}_{k-1},x\widetilde{G}_{k-1}],
  \widetilde{G}_k=\hf(\widetilde{P}_{k}),
\end{equation}
where $\widetilde{G}_0=G_0$.
Due to \cref{eq-hnf},  it is easy to check that
\begin{equation}\label{eq-lpf2}
 (\widetilde{P}_k)_{\Z}= (\widetilde{G}_k)_{\Z}=(F\cup\{x^iF\,|\,i=1,\ldots,k\})_{\Z}.
\end{equation}
\begin{remark}
Note that $\{x^iF\,|\,i=1,\ldots,k\}$ in \cref{eq-lpf2} are the
standard prolongation used in the XL algorithm \cite{xl1}
or a naive F4 style algorithm. The degree of $\widetilde{G}_k$ is $d+k$ which
increases with the loop number $k$, while the degree of ${G}_k$ in Algorithm {\em GHNF}$_1$
is always $d$, and this is the main advantage of our new prolongation.
A key idea in the F4 algorithm and the XL algorithm is that when $k$ is
large enough, a Gr\"obner basis of $F$ can be obtained
by doing Gaussian elimination to the coefficient matrix of $\widetilde{P}_k$.
We will prove that this is also true for the ``partial prolongation'' ${P}_k$ in Step 2 of the algorithm.
\end{remark}

Let $G_{k,s}$ and $\widetilde{G}_{k,s}$ be the sets of polynomials in
$G_{k}$ and $\widetilde{G}_{k}$ with degrees $\le s$, respectively.
Denote  $g_{k,j}$ and $\widetilde{g}_{k,j}$ to be
the polynomials in $G_{k}$ and $\widetilde{G}_{k}$ with degree
$j$, respectively. If there exist no such polynomials, $g_{k,j}$ and $\widetilde{g}_{k,j}$
are set to be zero. Clearly, $g_{k,d}\ne0$ and $\widetilde{g}_{k,d+i}\ne0$ for $i=0,\ldots,k$.

\begin{lemma}\label{claim-1}
We have $\lc(\widetilde{g}_{k,d})|\lc(\widetilde{g}_{k,d+1})|\cdots|\lc(\widetilde{g}_{k,d+k})$ and for $f\in (\widetilde{P}_{k+1})_{\Z}$ with $l=\deg(f)$, if $d< l\le d+k+1$ then $f\in (\widetilde{G}_{k,d},x\widetilde{G}_{k,l-1})_{\Z}$;
if $l\le d$ then $f\in (\widetilde{G}_{k,d},x\widetilde{G}_{k,d-1})_{\Z}$.
\end{lemma}
\begin{proof}
For convenience, denote $S_{i,k}=\widetilde{G}_{i,d}\cup x\widetilde{G}_{i,k}$
for $d-1\le k\le d+i-1$.
Since $\deg(S_{i,k}) = k+1$, $S_{i,k}\subset (\widetilde{G}_{i+1})_{{\Z}}$,
and $\widetilde{G}_{i+1}$ is a $\Z$-Gr\"obner basis, we have
$S_{i,k}\subset(\widetilde{G}_{i+1,k+1})_{\Z}.$

We prove the lemma by induction on the number of loops.
For  $k=0$,  since $xg_{0,d}$ is the only element in $\widetilde{P}_1$  with degree $d+1$, we have ${\lt(\widetilde{g}_{1,d+1})=\lt(x\widetilde{g}_{0,d})}$.
As a consequence,
if $f\in (\widetilde{P}_{1})_{\Z}$ and  $\deg(f)\le d$  then $f\in (S_{0,d-1})_{\Z}$.
If $f\in (\widetilde{P}_{1})_{\Z}$ and $\deg(f)= d+1$, then
it is obvious that $f\in (S_{0,d})_{\Z}=(\widetilde{P}_{1})_{\Z}$.
The lemma is proved for $k=0$.

Suppose the lemma is valid for $k \le i$.
By the induction hypothesis, since $\widetilde{g}_{i+1,j} \in \widetilde{P}_{i+1}$,
we have $\widetilde{g}_{i+1,j}\in (S_{i,j-1})_{\Z}$ if $d< j\le d+i+1$
and $\widetilde{g}_{i+1,j}\in (S_{i,d-1})_{\Z}$ if $j \le d$.
We first assume that $d< j\le d+i$. Since $x\widetilde{g}_{i,j-1}$ is the only
polynomial with degree $j$ in $S_{j-1}$, we have
\begin{equation}\label{eq-c11}
  \widetilde{g}_{i+1,j}=x\widetilde{g}_{i,j-1}+l_{i,j}
 \end{equation}
for some $l_{i,j}\in (S_{i,j-2})_{\Z}\subset (\widetilde{G}_{i+1,j-1})_{\Z}$.
Then, $\lc(\widetilde{g}_{i+1,j})=\lc(\widetilde{g}_{i,j-1})$,
and thus
 $\lc(\widetilde{g}_{i+1,j})|\lc(\widetilde{g}_{i+1,j+1})$ for $j=d+1,\ldots,d+i$
by the induction hypothesis.
Moreover, since $\widetilde{g}_{i+1,d}\in (S_{i,d-1})_{\Z}$ and
$\widetilde{g}_{i,d}$ and $x\widetilde{g}_{i,d-1}$ are the only
polynomials in $S_{i,d-1}$ with degree $d$,
we have $\lc(\widetilde{g}_{i+1,d})|\lc(\widetilde{g}_{i,d})$.
Then $\lc(\widetilde{g}_{i+1,d})|\lc(\widetilde{g}_{i+1,d+1})$ follows from
$\lc(\widetilde{g}_{i+1,d+1})=\lc(\widetilde{g}_{i,d})$.
The first part of the lemma is proved.

To prove the second part, we first show that if $d< q\le d+i+1$, then
\begin{equation}\label{eq-c12}
\widetilde{g}_{i+1,q} \in (\widetilde{G}_{i+1,q-1}, x\widetilde{G}_{i+1,q-1})_{\Z}.
\end{equation}
Since $\lc(\widetilde{g}_{i+1,q-1})|\lc(\widetilde{g}_{i+1,q})$,
$a = \frac{\lc(\widetilde{g}_{i+1,q})}{\lc(\widetilde{g}_{i+1,q-1})}$ is in $\Z$.
By \cref{eq-c11},
$\widetilde{g}_{i+1,q} - ax\widetilde{g}_{i+1,q-1} = x(\widetilde{g}_{i,q-1}  - ax\widetilde{g}_{i,q-2})+ l_{i,q} - axl_{i,q-1}$.
Since $\deg(\widetilde{g}_{i,q-1}  - ax\widetilde{g}_{i,q-2})\le q-1$,
we have  $\widetilde{g}_{i,q-1}- ax\widetilde{g}_{i,q-2} \in (\widetilde{G}_{i+1,q-1})_{\Z}$.
Also note $l_{i,j}\in (\widetilde{G}_{i+1, j-1})_{\Z}$.
Then \cref{eq-c12} is proved.
Let $f\in (\widetilde{P}_{i+2})_{\Z}=(\widetilde{G}_{i+1,d+i+1},x\widetilde{G}_{i+1,d+i+1})_{\Z}$ with $l=\deg(f)$.
Using \cref{eq-c12} repeatedly, we may assume
$f\in (\widetilde{G}_{i+1,d},x\widetilde{G}_{i+1,s})_{\Z}$ for some $s$.
Since $\deg(\widetilde{G}_{i+1,d}) = d$ and $\deg(x\widetilde{G}_{i+1,s}) = s+1$,
we have $s=l-1$ if $l > d$
and $s=d-1$ if $l\le d$, and  the lemma is proved.
\end{proof}

\begin{lemma}\label{claim-11}
We have $G_{k}=\widetilde{G}_{k,d}$ for any $k\ge0$.
\end{lemma}
\begin{proof}
This lemma is obviously valid for $k=0$.
Suppose it is valid for $k=i-1$, that is, $\widetilde{G}_{i-1,d}={G}_{i-1}$.
Since $\deg(G_i)\le d$, $G_i\subset (\widetilde{G}_i)_{\Z}$, and $\widetilde{G}_i$ is a $\Z$-Gr\"obner basis, we have $(G_i)_{\Z}=(P_i)_{\Z}\subset (\widetilde{G}_{i,d})_{\Z}$.
By \cref{claim-1} and the induction hypothesis,  we have $\widetilde{G}_{i,d}\subset (\widetilde{G}_{i-1,d},x\widetilde{G}_{i-1,d-1})_{\Z}=({G}_{i-1},x{G}_{i-1,d-1})_{\Z}=(P_i)_{\Z}$.
Hence,  $(G_i)_{\Z}= (\widetilde{G}_{i,d})_{\Z}$.
By Lemma \ref{lm-phnf}, $G_i$ and $\widetilde{G}_{i,d}$ are reduced $\Z$-Gr\"obner bases.
Hence  $G_i=\widetilde{G}_{i,d}$.
\end{proof}

\begin{lemma}\label{claim-12}
Suppose that Step 2 of Algorithm {\rm GHNF}$_1$ terminates at the $k$-th loop.
Then $(\widetilde{G}_{i})_{\Z} \subset (G_{k}$, $xg_{k,d},\ldots,x^{i}g_{k,d})_{\Z}$
for $i\ge 0$.
\end{lemma}
\begin{proof}
We have $G_{k}=G_{k+1}=\cdots$.
We prove the lemma by induction on $i$. The lemma is valid for $i=0$,
since $\widetilde{G}_{0}={G}_{0}\subset (G_{k})_{\Z}$.
Suppose that the lemma is valid for $i=t$.
From \cref{eq-lpf1}, $(\widetilde{G}_{t+1})_{\Z} = (\widetilde{G}_{t},x\widetilde{G}_{t})_{\Z}$.
By the induction hypothesis, $\widetilde{G}_{t}\subset  (G_{k}$, $xg_{k,d},\ldots,x^{t}g_{k,d})_{\Z}$. Then any $f\in\widetilde{G}_{t}$ can be written as
$f= f_0 + \sum_{j=0}^t c_j x^j g_{k,d}$, where $f_0\in G_{k,d-1}$ and $c_j\in\Z$.
Then
$xf=xf_0 + \sum_{j=0}^t c_i x^{i+1} g_{k,d}$.
Since $xf_0\in (xG_{k,d-1})_{\Z}\subset(G_{k+1})_{\Z}=(G_{k})_{\Z}$,
we have $xf\in(G_k,xg_{k,d},\ldots,x^{t+1}g_{k,d})_{\Z}$ and the lemma is proved.
\end{proof}

\begin{theorem}\label{th-co1}
Algorithm $\gHNF_1$ is correct. Furthermore,
Step 2 of Algorithm {\rm GHNF}$_1$ terminates in at most $D+d$ loops,
where $D=73d^5(h+\log d+1)$.
\end{theorem}
\begin{proof}
Suppose Step 2 of the algorithm terminates in the $k$-th loop.
Then, $G_{k}=G_{k+1}=\cdots$. We will show that $G_k$ is a Gr\"obner
basis of $(F)_{\Zx}$.
By \cref{eq-lpf2},  $(F)_{\Zx}= (G_k)_{\Zx} = (\widetilde{G}_k)_{\Zx}$.
To show that $G_k$ is a Gr\"obner basis, we will
prove that any $f\in (F)_{\Zx}$ can be reduced to zero by $G_k$.
By \cref{eq-lpf2}, there exists an integer $l$, such that $f\in (\widetilde{G}_l)_{\Z}$.
Since $(\widetilde{G}_i)_{\Z}\subset (\widetilde{G}_j)_{\Z}$ for $i < j$, we may assume that
$l \ge k$.
By \cref{claim-12}
$f\in (G_{k},xg_{k,d},\ldots,x^lg_{k,d})_{\Z}$.
Since $\{G_{k},xg_{k,d},\ldots,x^lg_{k,d}\}$ is a $\Z$-Gr\"obner basis,
we have $\overline{f}^{G_k}=0$ and $G_k$ is a Gr\"obner basis of $(F)_{\Zx}$.
Step 3 of the algorithm picks a reduced Gr\"obner basis, or the GHNF of $F$, from $G_k$.

We now prove the termination of the algorithm.
By Theorem \ref{degreeboundtranmat} and \cref{eq-lpf2},
$\widetilde{G}_D$ contains the GHNF of $F$ and hence a Gr\"obner basis of $(F)_{\Zx}$ by Theorem \ref{th-gg}.
By Lemma \ref{heightgeneral}, the reduced Gr\"obner basis of $(F)_{\Zx}$
has degree $\le d$.
By Lemma \ref{claim-11}, $G_D=\widetilde{G}_{D,d}$ contains the reduced Gr\"obner basis of $(F)_{\Zx}$.
From Example \ref{ex-11},  the termination condition may not be satisfied immediately
even if $G_i$ is a Gr\"obner basis of $(F)_{\Zx}$.
We will show that Step 2 will run  at most $d$ extra loops after
$G_k$ is a Gr\"obner basis.
Suppose $G_k=[g_{k,s_k},\ldots,g_{k,d}]$ is already a \gb~ of $(F)_{\Zx}$ for some $k\le D$ and
suppose $H_{k,1}=[g_{k,s_k},\ldots,g_{k,p}]$ such that $p$ is the maximal integer satisfying $g_{k,p}= g_{k+1,p}$.
Then, $H_{k,1}$ is also a Gr\"obner basis of $(F)_{\Zx}$.
If $p=d$, then, $H_{k,1}=G_k$, clearly $G_k= G_{k+1}$ and Step 2 terminates at $(k+1)$-th loop.
Otherwise, $p<d$ and $H_{k,1}\subset G_l$ for $l \ge k$.
Let $h_{k,p+1}$ be the reminder of $xg_{k,p}$ reduced by $H_{k,1}$ over $\Z$
and $H_{k,2}=[g_{k,s_k},\ldots,g_{k,p},h_{k,p+1}]$.
Then $\lt(h_{k,p+1})=\lt(xg_{k,p})$ and $\CM(H_{k,2})$ is an HNF.
Since $h_{k,p+1}$ is the minimal element in $(F)$ with degree $p+1$
and reduced w.r.t $H_{k,1}$, we have $g_{k+l,p+1}=h_{k,p+1}$ for $l>1$,
or equivalently  $H_{k,2}\subset G_l$ for $l\ge k+1$.
Similarly, we can prove that after each loop of Step 2, at least one
more element of $G_{l}$ will become stable.
As a consequence, Step 2 will terminate at most $D+d$ loops.
\end{proof}

\begin{theorem}\label{th-com1}
The bit size complexity of Algorithm {\rm GHNF}$_1$ is $O(d^{11+\theta+\varepsilon}(h+\log d)^{2+\varepsilon}+d^{7+\varepsilon}(h+\log d)$ $B(d^6(h+\log d)))$,
 where  $\varepsilon>0$ is any sufficiently small number.
\end{theorem}
\begin{proof}
The computationally dominant step of the algorithm is Step 2 and we will estimate the complexity of this step.
In the $k$-th loop of Step 2, we need to compute the HNF of
the coefficient matrix $C_k$ of $P_k$.
It is clear that $C_k$ is of size $(d+1)\times s$ for some $s\le 2d+1$.
Also note that the height of $C_k$ is the same as that of $\CM(G_k)$.
By Lemma \ref{claim-11} and \cref{eq-lpf2}, $\CM(G_k)$ is part of
the HNF of $\CM(\cup_{i=0}^k x^k F)$.
By Theorem \ref{th-hnfc}, the height of  $C_k$  is
$\le (k+d)(\frac{1}{2}\log (k+d)+h)\le h_1=(D+2d)(\frac{1}{2}\log (D+2d)+h)=O(d^5(h+\log d)^2)$,
since the loop will terminate at most $D+d$ steps.
Let $n=d+1,t=2d+1,r=d+1$, then the $\log \beta$ in Theorem \ref{th-hnfc} is $\log \beta=r(\frac{1}{2}\log r+h_1)=O(d^6(h+\log d))$.
To simplify the formula for the complexity bound,
we replace $O(\log^2(s)\log\log(s)\log\log\log(s))$ by $O(s^{\varepsilon})$ for
a sufficiently small number $\varepsilon$.
Hence, the complexity for each loop is
\begin{align*}
  &~~~~O(tnr^{\theta-2}(\log \beta)M(\log\log \beta)/\log\log \beta+kn\log rB(\log \beta))\\
  &\le O(d^{6+\theta+\varepsilon}(h+\log d)^{1+\varepsilon}+d^{2+\varepsilon}B(d^6(h+\log d))) \text{ for any }\varepsilon>0.
\end{align*}
By Theorem \ref{th-co1}, the number of loops is bounded by $D+d$.
So the worst complexity of the Algorithm {\rm GHNF}$_1$ is $(D+d)O(d^{6+\theta+\varepsilon}(h+\log d)^{1+\varepsilon}+d^{2+\varepsilon}B(d^6(h+\log d)))$ $=
O(d^{11+\theta+\varepsilon}(h+\log d)^{2+\varepsilon}+d^{7+\varepsilon}(h+\log d)B(d^6(h+\log d)))$.
\end{proof}

In Theorem \ref{th-com1}, setting $\theta=2.376$ and $\varepsilon=0.004$ and
noticing that $d^{7+\varepsilon}(h+\log d)B(d^2(h+d)))$ can be omitted now
comparing to the first term, we have

\begin{coro}\label{cor-com1}
The bit size complexity of Algorithm {\rm GHNF}$_1$ is $O(d^{13.38}(h+\log d)^{2.004})$.
\end{coro}

\begin{remark}\label{rem-com1}
The number $m$ in the input of Algorithm {\rm GHNF}$_1$ is not in the complexity bound.
The reason is that the size of the polynomial vector $P_k$ in Step 2
of the algorithm depends on $d$ only.
Only the complexity of Step 1 depends on $m$ and
by Theorem \ref{th-hnfc}, the complexity of
Step 1 is $O^{\sim}(md^{\theta+1}(h+d))$ which is
comparable to the complexity bound  in
Theorem \ref{th-com1} only when $m=O^{\sim}(d^{10})$. We therefore omit this term.
\end{remark}

Finally, we prove a property of the syzygy modules of $\Z[x]$ ideals, which will be used in the next section.
In Algorithm {\rm GHNF}$_1$, for any $k\ge 1$, let $v_{k-1}=\#(G_{k-1})$ be the number of columns of $G_{k-1}$. Then $u_k=\#(P_k)=2v_{k-1}-1$.
Let
$X_k=\left(
 \begin{array}{cccccc}
1& x &      &   &   &   \\
 &   &   \ddots &   &   &   \\
 &   &     & 1& x &   \\
 &   &     &   &   & 1 \\
 \end{array}
 \right)_{v_{k-1}\times u_k}$.
Then $P_k=G_{k-1}X_k = \X_d M_k$, where $M_k=\CM(P_k)$.
Let $[\0,H_k]=M_kU_k$ be the HNF of $M_k$, where $U_k=[U_{k,1},U_{k,2}]$ is a unimodular matrix satisfying $\0=M_kU_{k,1},~H_k=M_kU_{k,2}$.
By \cref{eq-ph1},
$$G_{k}=P_kU_{k,2}=FU_{0,2}X_1\cdots U_{k-1,2}X_kU_{k,2},\quad
  P_k=FU_{0,2}X_1\cdots U_{k-1,2}X_k,$$
where $G_0=\hf(F)=FU_{0,2}$.
For any $k\ge 1$, we define a map
\begin{align*}
\varphi_k: \Z[x]^{u_k}& \rightarrow \Z[x]^m\\
\u \quad & \mapsto U_{0,2}X_1\cdots U_{k-1,2}X_k\u.
\end{align*}
In particular, let $\varphi_0:  \Z[x]^m \rightarrow \Z[x]^m$ be the identity map.
The following result shows how to find a set of generators for the syzygy module $\Syz(F)$.

\begin{prop}\label{le-syz}
For any $\u\in \Syz(F)\subset\Zx^m$ and $\deg(\u)=l$, we have $\u\in (\bigcup_{k=0}^{l}\bigcup_{j=0}^{l-k}x^j\varphi_k(U_{k,1}))_{\Z}$.
Moreover, $\Syz(F) =(\bigcup_{k=0}^{d}\varphi_k(U_{k,1}))_{\Zx}$.
\end{prop}
\begin{proof}
By Theorem \ref{homogeneousZx}, $\Syz(F)$ can be generated by elements in $\Z[x]^m$ with degrees $\le d$.  We  need only to show the first statement.
Let $P_0=F$, $\u_0'=\u$.

Since $F\varphi_k(U_{k,1})=FU_{0,2}X_1\cdots U_{k-1,2}X_kU_{k,1}=
P_kU_{k,1}=\X_d M_kU_{k,1}=\0$ for any $k\ge 0$, we have $\varphi_k(U_{k,1})\subset \Syz(F)$.
By Lemma \ref{HNF over integer}, the lemma is valid for $l=0$.
If $l>0$, it  suffices to show that, for any $0\le q\le l$,
there exists a $\u_q' \in \Z[x]^{{u_q}}$ with $\deg(\u_q')\le l-q$,
such that $\u=\varphi_q(\u_q')\mod$
 $(\bigcup_{k=0}^{q-1}\bigcup_{j=0}^{l-k}x^j\varphi_k(U_{k,1}))_{\Z}$.
In this case, $P_q\u_q'=FU_{0,2}X_1\cdots U_{q-1,2}X_q\u_q'=F\u=0$.
It is valid for $q=0$. Suppose it is also valid for $q=i$.
  Let $\u_{i}'\in \Z[x]^{v_{i}'}$ with $\deg( \u_{i}')\le l-i$, such that $\u=\varphi_{i}(\u_{i}')\mod (\bigcup_{k=0}^{i-1}\bigcup_{j=0}^{l-k}x^j\varphi_k(U_{k,1}))_{\Z}$ and  $P_{i}\u_{i}'=0$.
  Let $\u_{i}''=U_{i}^{-1}\u_{i}'=[u_1,\ldots,u_{v_{i}'-v_i},0,\ldots,0]^\tau+[0,\ldots,0,
  u_{v_{i}'-v_i+1},\ldots,u_{v_{i}'}]^\tau$.
   Then, $\u_{i}'=U_{i}\u_i''=U_{i,1}[u_1,\ldots,u_{v_{i}'-v_i}]^\tau+U_{i,2}[u_{v_{i}'-v_i+1},\ldots,u_{v_{i}'}]^\tau$.
   Take $\u_{i}=[u_{v_{i}'-v_i+1},\ldots,u_{v_{i}'}]^\tau$.
   Then, $\u_i'=U_{i,2}\u_i\mod (\bigcup_{j=0}^{l-i}x^jU_{i,1})_{\Z}$, $G_{i}\u_{i}=P_{i}U_{i,2}\u_{i}=P_{i}\u_{i}'=0$.

For simplicity, denote $\u_i$ as $\u_i=[u_1,\ldots,u_{v_i}]^\tau$. Then $\deg(u_{v_{i}})\le l-i-1$ and $\deg(u_j)\le l-i$ for $1\le j<v_i$.
Let $u_j=u_{j,0}+p_jx$ for $1\le j< v_i$, where $u_{j,0}\in \Z$ and $p_j\in\Zx$ and $\deg(p_j)\le \deg(u_j)-1\le l-i-1$ .
Take $\u_{i+1}'=[u_{1,0}, p_1, \ldots,u_{v_i-1,0},  p_{v_i-1},u_{v_i}]^\tau$.
Then $\deg(\u_{i+1}')\le l-i-1$ and $\u_i=X_{i+1}\u_{i+1}'$.
Hence, $\u=\varphi_{i+1}(\u_{i+1}') \mod
 (\bigcup_{k=0}^{i}\bigcup_{j=0}^{l-k}x^j\varphi_k(U_{k,1}))_{\Z}$ and $P_{i+1}\u_{i+1}'=G_iX_{i+1}\u_{i+1}'=G_i\u_i=0$. The lemma is proved.
\end{proof}

\subsection{The $\Z[x]^n$ case}
\label{sec-algn}
In this section, an algorithm will be given to compute the \gHNFs for $\Zx$-lattices in $\Zx^n$, which is a generalization of Algorithm GHNF$_1$.

In this section, we assume $F=(f_{ij})_{n\times m}=[{\bf f}_1,\ldots,{\bf f}_m]\in \Zx^{n\times m}$ and  denote by $m=\#(F)$ to be the number of columns of  $F$.
Let $v_i=\max_{1\leq j \leq m}({\rm \deg}(f_{ij})),i=1,\ldots,n$, and
\begin{equation}\label{X2}
\X_F=\left(
       \begin{array}{ccccccccccccc}
         1 & x & \ldots & x^{v_1} &  &  &  & &  &  &  &  &  \\
          &  &  &  & 1 & x & \ldots & x^{v_2} &  & & &  &  \\
           &  &  &  &  &  &  &  & \ddots & & &  &  \\
          &  &  &  & &  &  &  &  & 1 & x & \ldots & x^{v_n} \\
       \end{array}
     \right)_{n\times s},
     \end{equation}
where $s=\sum_{i=1}^{n}(v_i+1)$.
Then, $F$ can be written in the matrix form: $F=\X_F C$,
where $C\in \Z^{s\times m}$ is called the {\em coefficient matrix} of $F$ and
is denoted by $C=\CM(F)$.
Let $[\0,H]=C[U_1,U_2]$ be the HNF of $C$, where $H$ has no zero columns
and $\0=CU_1$ and $H=CU_2$.
Then $F_1=\X_F H$ is called the PHNF of $F$ and is denoted by
\begin{equation}\label{eq-hfn1}
F_1=\hf(F) = \X_F H = \X_F CU_2 = FU_2.
\end{equation}

For a matrix $M\in\Zx^{n\times m}$,
denote by $M(\cdot,i)$ to be the $i$-th columns of $M$ and $M(i,\cdot)$ to be the $i$-th row of $M$. For  $\f\in\Zx^n$, denote by $\f(t)$ to be the polynomial in the $t$-th row of $\f$.
For $F=[{\bf f}_1,\ldots,{\bf f}_m]\in \Zx^{n\times m}$,
define the operation Divide as:
$$\hbox{Divide}(F)=(Q_1,\ldots,Q_n),$$
where either
$Q_t=[\f_{k_{t,1}},\ldots,\f_{k_{t,{s_t}}}]$ such that  $\f_{k_{t,i}}(t)\ne 0$,
and $\f_{k_{t,i}}(j)=0$ for $i=1,\ldots,s_t$ and $j>t$;
or $Q_t=\emptyset$ if such $\f_{k_{t,{s_t}}}$ do not exist.
Furthermore, it is always assumed that
$\deg(\f_{k_{t,1}}(t))\le\cdots\le\deg(\f_{k_{t,{s_t}}}(t))$.
For $d\in\N$, denote
$$Q_{t}^{(d)}=[\f_{k_{t,1}},\ldots,\f_{k_{t,s}}]$$
such that $\deg(\f_{k_{t,i}}(t)) \le d$ for $i=1,\ldots,s$ and
$\deg(\f_{k_{t,j}}(t)) > d$ for $j=s+1,\ldots,s_t$.
We now give the algorithm.

\begin{algorithm}[!htb]
\label{HNFn}
\caption{{\rm {\rm GHNF}$_n$$(F)$}}
\begin{algorithmic}[1]
\REQUIRE~~
$F\in \Zx^{n\times m}$ and with $d=\deg(F)$.
\ENSURE~~
 $G\in\Zx^{n\times s}$, which is the \gHNF of $F$.

 \STATE $G_0=\hf(F)$, $k=0$.
 \STATE(loop) \label{niteration}
 $k=k+1$;\\
 $(G_{k-1,1},\ldots,G_{k-1,n})=\Div(G_{k-1})$.\\
 $P_{k,t}=[G_{k-1,t}^{(d_{t})},xG_{k-1,t}^{(d_t-1)}],t=1,\ldots,n$, where $d_t=(n-t+1)d$.\\
 $P_k=[P_{k,1},\ldots,{P_{k,n}}]$.
 $G_k=\hf(P_k)$.\\
 If $G_k\ne G_{k-1}$, repeat Step 2.
 \STATE\label{npick}
 For $t$ from 1 to $n$, let $G_{k-1,t}=[\g_{k-1,1},\ldots,\g_{k-1,k_t}]$, $P_t=[\g_{k-1,1}]$;\\
 ~~for $j$ from $2$ to $k_t$,
 ~if {$\lc(\g_{k-1,j-1}(t))\ne \lc(\g_{k-1,j}(t))$},
 $P_t=P_t\cup \{\overline{\g_{k-1,j}}^{P_t}\}$.\\
 \STATE\label{nterminal} Return $G=[P_1,\ldots,P_n]$.
 \end{algorithmic}
 \end{algorithm}
Note that the number $d_t$ is from Theorem \ref{degreeheightboundforGHNF}.
We give the following illustrative example.

\begin{example}
Let  $F=\left(
  \begin{array}{cc}
  6x+1 & 3x\\
  2x  & 5x+1\\
  \end{array}
  \right)$. We have $d=1$.\\
  Step 1: $G_0=\hf(F)=\left(
  \begin{array}{cc}
  24x+5 & -9x-2 \\
  -2  & x+1 \\
  \end{array}
  \right).$\\
$1$-th loop: 
  $(G_{0,1},G_{0,2})={\rm Divide}(G_0)$, where\\
  $~~~~~~G_{0,1}=[~],~G_{0,2}= G_0$. Also, we have $d_1=2,d_2=1$.\\
  $~~~~~~P_{1,1}=[~], P_{1,2}=\left(
  \begin{array}{ccc}
  24x+5 & 24x^2+5x & -9x-2 \\
  -2  & -2x  & x+1 \\
  \end{array}
  \right).$\\
  $~~~~~~P_1=[P_{1,1},{P_{1,2}}], G_1=\hf(P_1)=\left(
  \begin{array}{ccc}
  24x^2+11x+1 & -24x-5 & -9x-2 \\
  0  & 2  & x+1 \\
  \end{array}
  \right).$\\
  $2$-th loop: 
   $(G_{1,1},G_{1,2})={\rm Divide}(G_1)$, where\\
    $~~~~~~G_{1,1}=\left(\begin{array}{c}
  24x^2+11x+1 \\
  0\\
  \end{array}
  \right),
  G_{1,2}=\left(\begin{array}{cc}
  -24x-5 & -9x-2 \\
  2  & x+1 \\
  \end{array}
  \right).$\\
  $~~~~~~P_{2,1}=\left(\begin{array}{c}
  24x^2+11x+1 \\
  0\\
  \end{array}
  \right),
  P_{2,2}=\left(\begin{array}{ccc}
  -24x-5 & -24x^2-5x & -9x-2 \\
  2  & 2x  & x+1 \\
  \end{array}
  \right).$\\
  $~~~~~~P_2=[{P_{2,1}},P_{2,2}], G_2=\hf(P_2)=\left(
  \begin{array}{ccc}
  24x^2+11x+1 & -24x-5 & -9x-2 \\
  0  & 2  & x+1 \\
  \end{array}
  \right).$\\
  $~~~~~~G_2=G_1$ and the loop terminates.\\
In Step 3, we can easily get the \gHNF of $F$: $G=G_2$.
\end{example}

Similar to GHNF$_1$, we consider the following
``full prolongation''
\begin{eqnarray}\label{nfullpro}
&&  \widetilde{P}_{k,t}=[\widetilde{G}_{k-1,t},x \widetilde{G}_{k-1,t}],t=1,\ldots,n,\nonumber\\
&&  \widetilde{P}_k=[\widetilde{P}_{k,1},\ldots,\widetilde{P}_{k,n}]
     =[\widetilde{G}_{k-1},x\widetilde{G}_{k-1}],\label{eq-fpn}\\
&&  \widetilde{G}_k=\hf(\widetilde{P}_k),~
  [\widetilde{G}_{k,1},\ldots,\widetilde{G}_{k,n}]={\rm Divide}(\widetilde{G}_k),\nonumber
\end{eqnarray}
where $\widetilde{G}_0=G_0$.
Due to \cref{eq-hnf},  it is easy to check that
\begin{equation}\label{eqn-lpf2}
 (\widetilde{G}_k)_{\Z}=(\widetilde{P}_k)_{\Z}=(F\cup\{x^iF\,|\,i=1,\ldots,k\})_{\Z}.
\end{equation}

We define a new monomial order as follows:
$x^{\alpha}\e_i\prec' x^{\beta}\e_j$ if and only if $\alpha <\beta$ or $\alpha=\beta$ and $i<j$.
Similar to the order $\prec$, the order $\prec'$ can be extended to the polynomial vectors of $\Z[x]^n$.
Moreover, the S-vector of $\f,\g\in \Z[x]^m$ is the same as \cref{eq-svec}.
A nice property of the order $\prec'$ is: if $\max(\deg(\f),\deg(\g))\le d$, then $\deg(S_{\prec'}(\f,\g))\le d$.
 We can easily obtain the following result.
 \begin{lemma}\label{le-syzgb}
   Let $F\in \Z[x]^{n\times m}$ and $d=\deg(F)$. Then $\Syz(F)$ has a \gb~
   with degree $\le nd$ $w.r.t. \prec'$.
 \end{lemma}
\begin{proof}
Let $S=\{\u\,|\, \u\in \Syz(F),~\deg(\u)\le nd\}$.
By Theorem \ref{homogeneousZx}, $S$ generates $\Syz(F)$.
Then, $S$ contains a \gb~ $G$ of $\Syz(F)$ w.r.t $\prec'$,
since the S-vector of any $\u,\v\in S$ w.r.t $\prec'$ is still in $S$.
\end{proof}

Let $F_{(t)}\in\Zx^{t\times m}$ be the last $t$ rows of $F$ and
\begin{equation}\label{eq-sk}
S_t=\{\u\in\Z[x]^m\,|\, \u\in \Syz(F_{(t)}), \deg(\u)\le td\}.
 \end{equation}
By Lemma \ref{le-syzgb}, $S_t$ contains a \gb~$G_t$ with $\deg(G_t)\le td$.
Then, for any $\u\in \Syz(F_{(t)})$ with $\deg(\u)\le k$, we have $\u\in (S_t,xS_t,\ldots,$ $x^{\max(0,~k-td)}S_t)_{\Z}$.
Moreover, we have {$(S_1)_{\Zx}\supseteq (S_2)_{\Zx}\supseteq \cdots\supseteq (S_n)_{\Zx}$}.

Let $u_{k,t}=\#(G_{k,t}^{(d_t)})$, $v_{k,t}=\#(G_{k,t}^{(d_t-1)})$,
 $w_{k,t}=\#(G_{k,t})$, and $r_{k,t} = u_{k-1,t} + v_{k-1,t} = \#(P_{k,t})$.
Define a matrix  $X_{k,t}=(x_{i,j})\in\Zx^{w_{k,t} \times r_{k,t}}$ as follows.
If $G_{k,t}=[~]$, then $X_{k,t}=[~]$.
Otherwise, $x_{i,i}=1$ for $i=1,\ldots,u_{k,t}$,
{$x_{i,u_{k,t}+i}=x$} for $i=1,\ldots,v_{k,t}$, and all other $x_{i,j}$ are zero.
Then, we have
\begin{equation}\label{eq-pkt}
P_{k,t}=G_{k-1,t}X_{k-1,t}
\end{equation}
for any $k$ and $t$.
Let $M_{k}=\CM(P_k)$ and $[\0,H_k]=M_kU_k$ the HNF of $M_k$.
From \cref{eq-hfn1}, we have $[\0,G_k]=P_kU_k$.

For each $k>0$, let $U_{k}$ be defined as above and $\widetilde{U}_{k,n}$ be the last $r_{k,n}$ rows of $U_k$.
We rewrite $\widetilde{U}_{k,n}$ as $\widetilde{U}_{k,n}=[V_{k,1},V_{k,2}]$, where $V_{k,1}$ consists of
the column vectors of $\widetilde{U}_{k,n}\cap \Syz(F_{(1)})$.
Let $Q_k=[P_{k,1},\ldots,P_{k,n-1}]$ and
$U_{k}=\left(\begin{array}{cc}
  W_{k,1} & W_{k,2} \\
  V_{k,1} & V_{k,2}\\
  \end{array}
  \right).$
From $[\0,G_k]=P_kU_k$, we have
\begin{eqnarray*}
&&[\0,G_{k,1},\ldots,G_{k,n-1}]=P_k
\left(\begin{array}{c}
  W_{k,1}  \\
  V_{k,1} \\
  \end{array}
  \right)
  =[Q_k,P_{k,n}]
\left(\begin{array}{c}
  W_{k,1}  \\
  V_{k,1} \\
  \end{array}
  \right) = Q_k W_{k,1} + P_{k,n}V_{k,1}. \\
&&G_{k,n}=
P_k
\left(\begin{array}{c}
  W_{k,2}  \\
  V_{k,2} \\
  \end{array}
  \right)
  =[Q_k,P_{k,n}]
\left(\begin{array}{c}
  W_{k,2}  \\
  V_{k,2} \\
  \end{array}
  \right)
  =Q_kW_{k,2}+P_{k,n}V_{k,2}.
 \end{eqnarray*}
From the above equations, we have
$G_{k,n}(n,\cdot)=P_{k,n}(n,\cdot)V_{k,2}$, since the elements in the last row of $Q_k$ are all 0.
Since $P_{k,n}V_{k,1}\in (P_k)_{\Z}=(G_k)_{\Z}$ and the last row of $P_{k,n}V_{k,1}$ is zero, we have

\begin{equation}\label{eq-sz0}
(P_{k,n}V_{k,1})_{\Z}\in (G_{k,1},\ldots,G_{k,n-1})_{\Z}.
\end{equation}
Similarly, $G_{k,n}-P_{k,n}V_{k,2}=Q_kW_{k,2} \in (G_{k,1},\ldots,G_{k,n-1})_{\Z}$,
that is, $G_{k,n}=P_{k,n}V_{k,2}\mod (G_{k,1},\ldots,G_{k,n-1})_{\Z}$.
Similar to the $\Zx$ case, for $k>0$,  we define a  map $\phi_k$:
\begin{align*}
  \phi_{k}: \Z[x]^{r_{k,n}}&\rightarrow \Z[x]^m\\
  \u &\mapsto V_{0,2}X_{1,n}\cdots V_{k-1,2}X_{k,n}\u,
\end{align*}
where $X_{k,n}$ is from \cref{eq-pkt}.
Let $P_{0,n}=F$, $r_{0,n}=m$ and $\phi_{0}: \Z[x]^{m}\rightarrow \Z[x]^m$ be the identity map in particular.
Thus, we have $$G_{k,n}(n,\cdot)=P_{k,n}(n,\cdot)V_{k,2}=F(n,\cdot)V_{0,2}X_{1,n}\cdots V_{k-1,2}X_{k,n}V_{k,2},$$
$$P_{k,n}(n,\cdot)=G_{k-1,n}(n,\cdot)X_{k-1,n}=F(n,\cdot)V_{0,2}X_{1,n}\cdots V_{k-1,2}X_{k,n}.$$
From \cref{eq-sz0}, we have
\begin{equation}\label{eq-sz1}
F\phi_k(V_{k,1})=FV_{0,2}X_{1,n}\cdots V_{k-1,2}X_{k,n}V_{k,1}=P_{k,n}V_{k,1} \subset
(G_{k,1},\ldots,G_{k,n-1})_{\Z}
 \end{equation}
for  each $k\ge 0$.
Hence, $\phi_k(V_{k,1})\subset \Syz(F_{(1)})$.

\begin{lemma}\label{le-syzn}
Let $F\in \Zx^{n\times m}$. For any $\u\in \Syz(F_{(1)})$ and $\deg(\u)=l>0$, we have $\u\in (\bigcup_{k=0}^{l}\bigcup_{j=0}^{l-k}x^j\phi_k(V_{k,1}))_{\Z}$ for $k>0$.
Moreover, if $l\le d$, we have $F\u\in (G_{l,1},\ldots,G_{l,n-1})_{\Z}$.
\end{lemma}
\begin{proof}
The proof of the first statement is similar to the proof of Proposition \ref{le-syz}.
Assume $l\le d$. We have $x^j (G_{k,1},\ldots,G_{k,n-1})_{\Z}\subset (G_{k+j,1},\ldots,G_{k+j,n-1})_{\Z}$ for any $j\le d-k$, by  our prolongation.
By \cref{eq-sz1}, we have $F\u\in (\bigcup_{k=0}^{l}\bigcup_{j=0}^{l-k}x^jF\phi_k(V_{k,1}))_{\Z}
\subset (\bigcup_{k=0}^{l}\bigcup_{j=0}^{l-k}x^j (G_{k,1},\ldots,G_{k,n-1})_{\Z})_{\Z}
\subset (G_{l,1},\ldots,G_{l,n-1})_{\Z}$.
\end{proof}

\begin{lemma}\label{le-st}
For any $1\le s\le n-1$, we have $G_{k,j}=\widetilde{G}_{k,j}$ for $k\le sd$ and
$1\le j\le n-s$.
\end{lemma}
\begin{proof}
First, let $s=1$.
$G_0=\widetilde{G}_0=FU_{0,2}$. Then, $G_{0,j}=\widetilde{G}_{0,j}$ for $1\le j\le n$.
This lemma is valid for $k=0$.
Suppose it is valid for $k=l<d$, $i.e.$, $G_{l,j}=\widetilde{G}_{l,j}$ for $1\le j\le n-1$.
We need to show $G_{l+1,j}=\widetilde{G}_{l+1,j}$ for $1\le j\le n-1$.
For any $\f\in (\widetilde{G}_{l+1,1},\ldots,\widetilde{G}_{l+1,n-1})_{\Z} \subset(\widetilde{P}_{l+1})_{\Z}=(F,xF,\ldots,x^{l+1}F)_{\Z}$,
there exists a $\u\in \Z[x]^m$, such that $\f=F\u$ with $\deg(\u)\le l+1$, and $\u\in \Syz(F_{(1)})$.
By Lemma \ref{le-syzn}, we have $\f=F\u\in (G_{l+1,1},\ldots,G_{l+1,n-1})_{\Z}$.
Thus, we have $G_{l+1,j}=\widetilde{G}_{l+1,j}$ for $1\le j\le n-1$, since $G_{l+1,j}\subset\widetilde{G}_{l+1,j}$ and both of them are reduced $\Z$-Gr\"obner bases.  The lemma is valid for $s=1$.

Suppose the lemma is valid for $s=p-1$. Then we have $G_{(p-1)d,j}=\widetilde{G}_{(p-1)d,j}$ for $1\le j\le n-p+1$.
By \cref{eq-sk} and \cref{eqn-lpf2}, $FS_{p-1}\subset (\widetilde{G}_{(p-1)d,1},\ldots,\widetilde{G}_{(p-1)d,n-p+1})_{\Z}=(F')_{\Z}$, where $F'=[G_{(p-1)d,1},\ldots,G_{(p-1)d,n-p+1}]$.

When $s=p$, for any $(p-1)d<k\le pd$ and $\f\in (\widetilde{G}_{k,1},\ldots,\widetilde{G}_{k,n-p})_{\Z}\subset (\widetilde{P}_k)_{\Z}$,
  there exists a $\u\in \Z[x]^m$ with $\deg(\u)\le k$, such that $\f=F\u$ and $\u\in \Syz(F_{(p)})\subset \Syz(F_{(p-1)})$.
  By Lemma \ref{le-syzgb}, $\u\in (S_{p-1})_{\Zx} $ and $\u\in (S_{p-1},\ldots,x^{k-(p-1)d}S_{p-1})_{\Z}$.
  Then, $\f=F\u\in (F',\ldots,x^{k-(p-1)d}F')_{\Z}$.
  Hence we have  $\f = F' \v$ for some $\v\in\Syz(F'_{(p)})$ with $\deg(\v)\le k-(p-1)d\le d$ and  $F'_{(p)}$ being the last $p$ rows of $F'$.
Since the last $p-1$ rows of $F'$ are all zeros, it can be reduced to the $s=1$ case.
Considering the algorithm {\rm {\rm GHNF}$_{n}$$(F')$} and the analysis for the $s=1$ case, we have $\f=F\v' \in (G_{k,1},\ldots,G_{k,n-p})_{\Z}$.
Thus, $G_{k,j}=\widetilde{G}_{k,j}$ for $1\le j\le n-p$.
\end{proof}

The following lemma asserts that the last $s$ rows of $\widetilde{P}_{k}$ do not contribute
to the  first $(n-s)$ rows of $\widetilde{G}_{k}$ for $k> sd$.
\begin{lemma}\label{le-st1}
Let $R=[G_{sd,1},\ldots,G_{sd,n-s}]$.
Then we have $\widetilde{G}_{k,n-s}\subset (R)_{\Zx}$ for $1\le s\le n-1$ and $k> sd$.
In particular,  $\widetilde{G}_{k,n-s}\subset (R,xR,\ldots,x^{k-sd}R)_{\Z}\subset (\widetilde{P}_{k,1},\ldots,\widetilde{P}_{k,n-s})_{\Z}$ for $1\le s\le n-1$ and $k> sd$.
\end{lemma}
\begin{proof}
Let $k> sd$. For any $\f\in \widetilde{G}_{k,n-s}\subset
 (\widetilde{P}_{k})_{\Z}$, there exists a $\u\in \Syz(F_{(t)})$ with $\deg(\u)\le k$,
such that $\f=F\u$.
By Theorem \ref{homogeneousZx}, $\u\in (S_s)_{\Zx}$.
By Lemma \ref{le-syzgb}, $\u\in (S_s,\ldots,x^{k-sd}S_s)_{\Z}$.
By Lemma \ref{le-st}, $G_{sd,j}=\widetilde{G}_{sd,j}$ for $1\le j\le n-s$, $1\le s< n$.
Then, By \cref{eq-sk} and \cref{eqn-lpf2}, $FS_s\subset (\widetilde{G}_{sd,1},\ldots,\widetilde{G}_{sd,n-s})_{\Z}=(R)_{\Z}$.
Thus, $\f=F\u\subset  (R,xR,\ldots,x^{k-sd}R)_{\Z}\subset (R)_{\Zx}$.

To show the second statement, first, let $k=sd+1$.
We have $\f \in (R,xR)_{\Z}=(\widetilde{P}_{td+1,1},\ldots,\widetilde{P}_{sd+1,n-s})_{\Z}$.
The lemma is valid for $k=sd+1$.
Suppose the lemma is valid for $k=l > sd$. Then, $\widetilde{G}_{l,n-s}\subset
(R,xR,$ $\ldots,x^{l-sd}R)_{\Z}\subset (\widetilde{P}_{l,1},\ldots,\widetilde{P}_{l,n-s})_{\Z}$.
We need to show $\widetilde{G}_{l+1,n-s}\subset
(\widetilde{P}_{l+1,1},\ldots,$ $\widetilde{P}_{l+1,n-s})_{\Z}$.
For any $\f\in \widetilde{G}_{l+1,n-s}$,
we have $\f\in (R,$ $xR,\ldots,x^{l-sd+1}R)_{\Z}=
((R,xR,\ldots,x^{l-sd}R)\cup x(R,$ $xR,\ldots,x^{l-sd}R))_{\Z}\subset (\widetilde{G}_{l,1},$ $\ldots,\widetilde{G}_{l,n-s},x\widetilde{G}_{l,1},$ $\ldots,x\widetilde{G}_{l,n-s})_{\Z}
=(\widetilde{P}_{l+1,1},\ldots,\widetilde{P}_{l+1,n-s})_{\Z}$.
The lemma is also valid for $k=l+1$.
\end{proof}

 \begin{lemma}\label{le-sn}
For any  $k\ge 1$ and $1\le t\le m$, let $R_{k,t}=[\widetilde{G}_{k-1,t}^{(d_t)},x\widetilde{G}_{k-1,t}^{(\widetilde{p}_{k-1,t}-1)}]$,
where $\widetilde{p}_{k-1,t}=$\newline
 $\max(d_t,$ $\max_{\g\in\widetilde{G}_{k-1,t}} \deg(\g(t)))$.
Then
we have  $\f\in (R_{k,1},\ldots,R_{k,n-s})_{\Z}$ whenever $\f=[f_1,\ldots,f_{n-s},0,\ldots,0]^\tau\in (\widetilde{P}_{k,1},\ldots,\widetilde{P}_{k,n-s})_{\Z}$.
 \end{lemma}
\begin{proof}
First, let $s=n-1$. If $k\le (n-1)d$, by Lemma \ref{le-st}, we have  $R_{k,1}=\widetilde{P}_{k,1}$.
 Then, $\f\in (\widetilde{P}_{k,1})_{\Z}=(R_{k,1})_{\Z}$.
 Otherwise, $k>(n-1)d$,  by Lemma \ref{le-st1}, $\f\in (\widetilde{P}_{k,1})_{\Z}
 \subset (G_{(n-1)d,1})_{\Zx}$.
 By Lemma \ref{claim-1}, $(\widetilde{P}_{k,1})_{\Z}=(R_{k,1})_{\Z}$.
 The lemma is valid for $s=n-1$.

Suppose the lemma is valid for $s=l+1\le n-1$, $i.e.$
for any $k>0$ and $\f\in (\widetilde{P}_{k,1},\ldots,\widetilde{P}_{k,n-l-1})_{\Z}$, $\f\in (R_{k,1},\ldots,R_{k,n-l-1})_{\Z}$.
Let $s=l$, $\f=[f_1,\ldots,f_{n-l},0,\ldots,$ $0]^\tau\in (\widetilde{P}_{k,1},\ldots,\widetilde{P}_{k,n-l})_{\Z}$.
If $k\le ld$, then, $R_{k,j}=\widetilde{P}_{k,j}$ for $1\le j\le n-l$.
  Thus, $\f\in (R_{k,1}, \ldots,R_{k,n-l})_{\Z}$.
   Otherwise, $k>ld$.
 If $f_{n-l}=0$, $\f\in (\widetilde{G}_{k,1},\ldots,\widetilde{G}_{k,n-l-1})_{\Z}$.
In this case, if $k\le (l+1)d$, $R_{k,j}=\widetilde{P}_{k,j} = P_{k,j}$ for $1\le j\le n-l-1$ by  Lemma \ref{le-st}.
  $\f\in(\widetilde{P}_{k,1},\ldots,\widetilde{P}_{k,n-l})_{\Z}=
  (R_{k,1},\ldots, R_{k,n-l-1},\widetilde{P}_{k,n-l})_{\Z}\subset
  (R_{k,1},\ldots,R_{k,n-l})_{\Z}$ by Lemmas \ref{claim-1} and \ref{le-syzn}.
  If $k>(l+1)d$,  by  Lemma \ref{le-st1}, $\f\in (\widetilde{P}_{k,1},\ldots,\widetilde{P}_{k,n-l-1})_{\Z}$.
  By the induction hypothesis,  $\f\in (R_{k,1},\ldots,R_{k,n-l-1})_{\Z}$.
  If  $f_{n-l}\ne 0$, by Lemma \ref{le-st1} we have
 $\f\in (\widetilde{P}_{k,1},\ldots,$ $\widetilde{P}_{k,n-l})_{\Z}
 \subset (G_{ld,1},\ldots,G_{ld,n-l})_{\Zx}$.
  Then, for $k>ld$ we have  $\f\in (\widetilde{P}_{k,1},\ldots,\widetilde{P}_{k,n-l-1},R_{k,n-l})_{\Z}$
   by Lemmas \ref{claim-1} and \ref{le-syzn}.
  Thus, by induction, $\f\in (R_{k,1},\ldots,R_{k,n-l})_{\Z}$.
  The lemma is proved.
 \end{proof}

\begin{lemma}\label{le-loopnequivalent}
We have $G_{k,t}^{(d_t)}(t,\cdot)=\widetilde{G}_{k,t}^{(d_t)}(t,\cdot)$
for any $k\ge 0$, $1\le t\le n$.
\end{lemma}
\begin{proof}
Note that $d_n = d$ and for the $n$-th row of $F$,
Algorithms GHNF$_n$ and Algorithm GHNF$_1$ are exactly the same.
Hence,  by Lemma \ref{claim-11},
we have $G_{k,n}^{(d_n)}(n,\cdot)=\widetilde{G}_{k,n}^{(d_n)}(n,\cdot)$ for any $k\ge 0$.
Set $s=n-t$ in Lemma \ref{le-st}, we have
$G_{k,j}=\widetilde{G}_{k,j}$ for any $1\le t\le n-1$,  $k\le (n-t)d$, and $1\le j\le t$.
We thus proved the lemma when $k\le (n-t)d$.
Set $s=n-t$ in Lemmas \ref{le-st1} and \ref{le-sn}, we have
$\widetilde{G}_{k,t}\subset (\widetilde{P}_{k,1},\ldots,\widetilde{P}_{k,t})_{\Z}
\subset (R_{k,1},\ldots,R_{k,t})_{\Z}$ for $1\le t\le n-1$ and $k>(n-t)d$.
Note that Lemma \ref{le-sn} is the analog of Lemma \ref{claim-1} in the case of $n>1$. Thus, similar to  Lemma \ref{claim-11}, we can prove $G_{k,t}^{(d_t)}(t,\cdot)=\widetilde{G}_{k,t}^{(d_t)}(t,\cdot)$ for $k > (n-t)d$. The lemma is proved.
\end{proof}

\begin{lemma}\label{claim-n2}
Suppose Step 2 of Algorithm {\rm GHNF}$_n$ terminates at the $k$-th loop
and let $\g_{k,t,d_t}$ be the last column vector of $G_{k,t}^{(d_t)}$.
Then $\deg(\g_{k,t,d_t})=d_t$ and for any $i\ge 0$,  $(\widetilde{G}_i)_{\Z}\subset (H_{i,1},\ldots,H_{i,n})_{\Z}$, where
   $H_{i,t}=(G_{k,t}^{(d_t)},x\g_{k,t,d_t},\ldots,x^{\max(i,k)-(n-t)d}\g_{k,t,d_t})$.
\end{lemma}
\begin{proof}
It is suffice to show $\widetilde{G}_{i,t}\subset (H_{i,1},\ldots,H_{i,t})_{\Z}$ for any $i\ge 0$ and $1\le t\le n$.
If $\deg(G_{k-1,t})< d_t$, then $\deg(G_{k,t}) \ge \deg(P_{k,t})> \deg(G_{k-1,t})$
and the algorithm does not terminate. Therefore, if $G_{k,t}\ne \emptyset$, then we have
$k\ge d_t-d=(n-t)d$ and hence $\deg(\g_{k,t,d_t})=d_t$.

First, let $t=1$.
Clearly, for any $i\le (n-1)d$, $\widetilde{G}_{i,1}=G_{i,1}\subset (G_{k,1}^{(d_1)})_{\Z}$, where  $=$ is based on Lemma \ref{le-st} and $\subset$ is valid because $(G_{j,1}^{(d_1)})_{\Z}\subset (G_{j+1,1}^{(d_1)})_{\Z}$ for any $j\ge 0$.
Thus, we have $\widetilde{G}_{(n-1)d,1}=G_{(n-1)d,1}\subset (G_{k,1}^{(d_1)})_{\Z}\subset (H_{(n-1)d,1})_{\Z}$.
Suppose it is valid for $i=j>(n-1)d$.
From  (\ref{eq-fpn}) and Lemma \ref{le-st1}, $(\widetilde{G}_{j+1,1})_{\Z}
=(\widetilde{G}_{j,1},x\widetilde{G}_{j,1})_{\Z}$.
By induction hypothesis, $\widetilde{G}_{j,1}\subset
(H_{j,1})_{\Z}$ where $H_{j,1}=(G_{k,1}^{(d_1)},x\g_{k,1,d_1},\ldots,
x^{\max(j,k)-(n-1)d}\g_{k,1,d_1})_{\Z}$.
Then, any $\g\in \widetilde{G}_{j,1}$ can be written as
$\g=\g_0+\sum_{l=0}^{\max(j,k)-(n-1)d}c_lx^l\g_{k,1,d_1}$,
 where $\g_0\in G_{k,1}^{(d_1-1)}$ and $c_l\in \Z$.
 Since $x\g_0\in (xG_{k,1}^{(d_1-1)})_{\Z}\subset (G_{k+1,1}^{(d_1)})_{\Z}
 =(G_{k,1}^{(d_1)})_{\Z}$,
 we have $(\widetilde{G}_{j+1,1})_{\Z}\subset (G_{k,1}^{(d_1)},x\g_{k,1,d_1},\ldots,x^{\max(j+1,k)-(n-1)d}\g_{k,1,d_1})_{\Z}$.
The lemma is valid for any $i\ge 0$ and $t=1$.

Suppose the lemma is valid for any $i\ge 0$ and $t\le s<n$.
 Then $(G_{j,1},\ldots,G_{j,s})_{\Z}\subset
 (\widetilde{G}_{j,1},\ldots,\widetilde{G}_{j,s})_{\Z}
 \subset (H_{j,1},\ldots,H_{j,s})_{\Z}$ for any $j\ge 0$.

 By induction, $(\widetilde{G}_{i,1},\ldots,
 \widetilde{G}_{i,s+1})_{\Z}=(G_{i,1},\ldots,
 G_{i,s+1})_{\Z}\subset (H_{i,1},\ldots,$ $H_{i,s},G_{i,s+1})_{\Z}$
 for $i\le (n-s-1)d$.
 Moreover,
$(G_{i,s+1}^{(d_{s+1})})_{\Z}\subset (G_{i+1,1},\ldots,G_{i+1,s},G_{i+1,s+1}^{(d_{s+1})})_{\Z}
\subset (H_{i+1,1},\ldots,H_{i+1,s},G_{i+1,s+1}^{(d_{s+1})})_{\Z}$ for any $i\ge 0$.
Since $d+i\le d_{s+1}$ and $H_{j,t}=H_{k,t}$ for any $j\le k$ and $1\le t\le n$, we have $(\widetilde{G}_{i,1},\ldots,
 \widetilde{G}_{i,s+1})_{\Z}$ $\subset (H_{k,1},\ldots,H_{k,s},G_{k,s+1}^{(d_{s+1})})_{\Z}
 \subset(H_{i,1},\ldots,H_{i,s},H_{i,s+1})_{\Z}$
 and the lemma is valid for $i\le (n-s-1)d$.

Suppose the lemma is valid for $i=j>(n-s-1)d$.
From (\ref{eq-fpn}), $(\widetilde{G}_{j+1,s+1})_{\Z}=(\widetilde{G}_{j,s+1},
x\widetilde{G}_{j,s+1})_{\Z}$.
By the induction hypothesis, $\widetilde{G}_{j,s+1}\subset (H_{j,1},\ldots,H_{j,s+1})_{\Z}$.
Then, any $\g\in \widetilde{G}_{j,s+1}$ can be written as
$\g=\sum_{t=1}^{s+1}(\g_{t,0}+\sum_{l=0}^{\max(j,k)-(n-t)d} c_{t,l}x^l\g_{k,t,d_{t}})$, where $\g_{t,0}\in G_{k,t}^{(d_t-1)},$ and
$c_{t,l}\in \Z$.
Moreover, since for any $i\ge 0$ and $t\le s+1$, $(G_{i,t}^{(d_{t})})_{\Z}\subset (G_{i+1,1},\ldots,G_{i+1,t-1},G_{i+1,t}^{(d_{t})})_{\Z}
\subset (H_{i+1,1},\ldots,H_{i+1,t-1},G_{i+1,t}^{(d_{t})})_{\Z}$,
we have
$x\g_{t,0}\in$ $ (G_{k+1,1},\ldots,$ $G_{k+1,t-1},G_{k+1,t}^{(d_t)})_{\Z}
\subset (H_{k+1,1},\ldots,H_{k+1,t-1}, G_{k+1,t}^{(d_t)})_{\Z}
=(H_{k+1,1},\ldots,H_{k+1,t-1}, G_{k,t}^{(d_t)})_{\Z}$.
Then,  $(\widetilde{G}_{j+1,s+1})_{\Z}\subset (H_{k+1,1},
\ldots,H_{k+1,s+1})_{\Z}$.
Since $\deg(\widetilde{G}_{j+1,s+1})\le d+j+1$, we have
$(\widetilde{G}_{j+1,s+1})_{\Z}\subset (H_{j+1,1},\ldots,H_{j+1,s+1})_{\Z}$.
\end{proof}

Notice that in the proof of Lemma \ref{claim-n2}, we need only $G_{k,t}^{(d_t)}=G_{k+1,t}^{(d_t)}$ for $1\le t\le n$. Then, we have the following corollary.
\begin{coro}\label{coro-n2}
  In the Algorithm {\rm GHNF}$_n$, if $G_{k,t}^{(d_t)}=G_{k+1,t}^{(d_t)}$ for $1\le t\le s$ for some positive integer $s\le n$, then $(\widetilde{G}_{i,s})_{\Z}\in (H_{i,1},\ldots,H_{i,s})_{\Z}$, where $H_{i,t}=(G_{k,t}^{(d_t)},x\g_{k,t,d_t},\ldots,
  x^{\max(i,k)-(n-t)d}\g_{k,t,d_t})_{\Z}$ for any $i\ge 0$, $1\le t\le s$.
\end{coro}

By this result, we obtain an equivalent termination condition for the Algorithm
GHNF$_n$:
\begin{lemma}\label{le-n2}
  In the Algorithm {\rm GHNF}$_n$, $G_{k}=G_{k+1}$ is equivalent to
  $G_{k,t}(t,\cdot)=G_{k+1,t}(t,\cdot)$ for $1\le t\le n$.
\end{lemma}
\begin{proof}
  Clearly, if $G_{k}=G_{k+1}$, we have $G_{k,t}(t,\cdot)=G_{k+1,t}(t,\cdot)$ for $1\le t\le n$.
  We just need to show the opposite direction.
  In this condition, we prove $G_{k,t}=G_{k+1,t}$ by  induction on $t$.
  Since $G_{j,1}(1,\cdot)=G_{j,1}$ for any $j$, the lemma is valid for $t=1$.
  Suppose $G_{k,t}=G_{k+1,t}$ for $1\le t\le s<n$.
 Since  $G_{k,t}(t,\cdot)=G_{k+1,t}(t,\cdot)$ for $1\le t\le n$,
 for any $\g'\in G_{k+1,s+1}$,  there exist a $\g\in G_{k,s+1}$
 satisfying $\g(s+1)=\g'(s+1)$.
  If $\g\in G_{k,s+1}^{(d_{s+1})}$, we have $\g\in (G_{k+1})_{\Z}$.
  Then, $\g-\g'\in (G_{k+1})_{\Z}$.
  Since $(\g-\g')(t)=0$ for $s+1\le t\le n$, we have $\g-\g'\in (G_{k+1,1},\ldots,G_{k+1,s})_{\Z}=(G_{k,1},\ldots,G_{k,s})_{\Z}$.
  Thus, $\g'\in (G_{k,1},\ldots,G_{k,s+1}^{(d_{s+1})})_{\Z}$ and
  $(G_{k+1,1},\ldots,G_{k+1,s},G_{k+1,s+1}^{(d_{s+1})})_{\Z}=
  (G_{k,1},\ldots,G_{k,s},G_{k,s+1}^{(d_{s+1})})_{\Z}$.
  Then, $G_{k,s+1}^{(d_{s+1})}=G_{k+1,s+1}^{(d_{s+1})}$ since both of them are
   reduced $\Z$-Gr\"obner bases.
  If $\g\notin G_{k,s+1}^{(d_{s+1})}$,  we have $\g\in (G_{k,1},\ldots,G_{k,s},$ $G_{k,s+1}^{(d_{s+1})},x\g_{k,s+1,d_{s+1}},
  \ldots,x^l\g_{k,s+1,d_{s+1}})_{\Z}$ for some $l\ge 0$ by Corollary \ref{coro-n2}.
  So is $\g'$ since $G_{k,s+1}^{(d_{s+1})}=G_{k+1,s+1}^{(d_{s+1})}$.
  Thus we have $\g-\g'\in (G_{k,1},\ldots,G_{k,s})$ since $(G_{k,1},\ldots,G_{k,s},$ $G_{k,s+1}^{(d_{s+1})},x\g_{k,s+1,d_{s+1}},$
  $\ldots,$ $x^l\g_{k,s+1,d_{s+1}})_{\Z}$ is a $\Z$-\gb.
  Then $(G_{k,1},\ldots,G_{k,s+1})_{\Z}=(G_{k+1,1},\ldots,G_{k+1,s+1})_{\Z}$.
  Since both of them are reduced $\Z$-Gr\"obner bases, we have $G_{k,s+1}=G_{k+1,s+1}$ .
\end{proof}

We now show the correctness of the algorithm.
\begin{theorem}\label{nloopnum}
Algorithm {\rm GHNF}$_n$ is correct. Furthermore,
Step 2 of Algorithm {\rm GHNF}$_n$ terminates in at most $D+nd$ loops,
where $D=73n^8d^5(h+\log (n^2d)+1)$.
\end{theorem}
\begin{proof}
Suppose Step 2 of the algorithm terminates in the $k$-th loop.
The fact that $G_k$ is a Gr\"obner basis of $(F)_{\Zx}$ can be proved
similarly to that of Theorem \ref{th-co1}, where
instead of Lemma \ref{claim-12}, we use Lemma \ref{claim-n2}.

We now prove the termination of the algorithm.
By Theorem \ref{degreeboundtranmat} and \cref{eqn-lpf2},
$\widetilde{G}_D$ contains the GHNF of $F$ and hence a Gr\"obner basis of $(F)_{\Zx}$ by  Theorem \ref{th-gg}.
By Lemma \ref{heightgeneral}, if ${\mathcal{C}}$ is the GHNF of $F$ and has form \cref{reducedGB},
then $\deg({\mathcal{C}}(r_i,\cdot))\le d_{r_i}=(n-r_i+1)d,i=1,\ldots,t$.
Hence,  $G_D$ also contains a Gr\"obner basis of $(F)_{\Zx}$ by Lemma \ref{le-loopnequivalent}.
Similar to the $\Zx$ case, the termination condition may not be satisfied immediately
even if $G_i$ is a Gr\"obner basis of $(F)_{\Zx}$.
By Lemma \ref{le-n2}, Algorithm GHNF$_n$ terminates at the $(k+1)$-th loop if and only if $G_{k,t}(t,\cdot)=G_{k+1,t}(t,\cdot)$ for $1\le t\le n$.
By Lemma  \ref{le-st1} and Lemma \ref{le-loopnequivalent}, after the $nd$-th loop,
$\deg(G_{i,t}(t,\cdot))=d_t$ and the computation of $G_{i,t}(t,\cdot)$ only depends on $G_{i,t}(t,\cdot)$ for $1\le t\le n$.
Also note that if $G_{i}$ is a Gr\"obner basis, then $G_{i,t}$ is either empty or
a Gr\"obner basis.
Then, similar to the proof of Theorem \ref{th-co1},
we can show that after $D$-loop, $G_{i,t}(t,\cdot)$ are
Gr\"obner bases for $t=1,\ldots,n$ and
after that the loop terminates for at most $d_1=dn$ extra steps.
\end{proof}

\begin{theorem}
The worst bit size complexity of Algorithm {\rm GHNF}$_n$ is $O(n^{26+2\theta+\varepsilon}$ $d^{15+\theta+\varepsilon}(h+\log (n^2d))^{4+\varepsilon}$ $+n^{19}d^{11}(h+\log (n^2d))^{2}\log (n^2d)B(n^{11}d^6(h+\log (n^2d))^2))$, where $h=\height(F)$ and $\varepsilon>0$ is a sufficiently small number.
\end{theorem}
\begin{proof}
In the $k$-th loop in Step 2, we need to compute the HNF of an integer matrix $M_k$
whose size is $n(d+k+1)\times s$, where $s\le (2d+1)+(4d+1)+\cdots+(2nd+1)=n(n+1)d+n$.
By  Theorems \ref{th-hnfc}, \ref{nloopnum}, and \cref{eqn-lpf2},
the
height of $M_k$ $\le n(D+nd+1)(\frac{1}{2}\log (n(D+nd+1))+h)=O(n^9d^5(h+\log(n^2d))^2):=h_2$.
The $\log \beta$ in  Theorem \ref{th-hnfc} can be taken as $\log \beta=(n(n+1)d+n)(\frac{1}{2}\log(n(n+1)d+n)+h_2)=O(n^{11}d^{6}(h+\log (n^2d))^2)$.
 To simplify the formula for the complexity bound,
we replace $O(\log^2(s)\log\log(s)\log\log\log(s))$ by $O(s^{\varepsilon})$ for
an sufficiently small number $\varepsilon$.
 The complexity in the $k$-th loop is
   $O(n(d+k+1)\cdot (n(n+1)d+n)^{\theta-1}(\log \beta) M(\log\log \beta)/(\log \log \beta)+n(d+k+1)\cdot(n(n+1)d+n)\log(n(n+1)d+n)B(\log \beta)) =(d+k+1)O(n^{10+2\theta+\varepsilon}d^{5+\theta+\varepsilon}(h+\log (n^2d))^{2+\varepsilon}+n^3d\log (n^2d)B(n^{11}d^6(h+\log (n^2d))^2)),
 $
  for any $\varepsilon>0$.
  Hence the total complexity is $\sum_{k=0}^{D+nd}
    (d+k+1)O(n^{10+2\theta+\varepsilon}$ $d^{5+\theta+\varepsilon}(h+\log (n^2d))^{2+\varepsilon}+n^3d\log (n^2d)B(n^{11}d^6(h+\log (n^2d))^2))
    =O(n^{26+2\theta+\varepsilon}d^{15+\theta+\varepsilon}$ $(h+\log (n^2d))^{4+\varepsilon}+n^{19}d^{11}(h+\log (n^2d))^{2}\log (n^2d)B(n^{11}d^6(h+\log n^2d)^2)).$
 %
\end{proof}

Similar to Corollary  \ref{cor-com1}, by setting
$\theta=2.376$ and $\varepsilon=0.001$, we have
\begin{coro}
The worst  bit size complexity of Algorithm {\rm GHNF}$_n$ is $O(n^{30.753}$ $d^{17.377}(h+\log (n^2d))^{4.001})$.
\end{coro}

Similar to  Remark \ref{rem-com1}, the number $m$ in the input is omitted in the complexity bound.

\section{Experimental results}

The algorithms presented in Section 4 have been implemented in  both Maple 18
and Magma 2.21-7. The timings given in this section are collected on a
PC with Intel(R) Xeon(R) CPU E7-4809 with 1.90GHz.
For each set of inpute parameters, we use the average timing of ten experiments
for random polynomials with coefficients between $[-100,100]$.

Table \ref{Z[x]eps} shows the timings of the Algorithm {\rm GHNF}$_1$ in Magma 2.21-7 and Maple 18,
and that of the Gr{\"o}bnerBasis command in Magma 2.21-7.
From  Theorem \ref{th-com1}, the degree of the input polynomials is
the dominant factor in the computational complexity of the algorithm.
In the experiments, the length of the input polynomial vectors is fixed to be 3.
The degrees are in the range $[45,80]$.

From the figure, we have the following observations.
The new algorithm is much more efficient than
the Gr{\"o}bnerBasis algorithm  in Magma.
As far as we know, the Gr{\"o}bnerBasis algorithm in Magma also
uses an  F4 style algorithm to compute the Gr{\"o}bner basis
and is also based on the computation of HNF of the coefficient matrices.
In other words, the Gr{\"o}bnerBasis algorithm in Magma
is quite similar to our algorithm and the comparison is fair.
The reason for Algorithm {\rm GHNF}$_1$ to be more efficient
is due to the way how the prolongation is done in Step 2 of algorithm {\rm GHNF}$_1$.
By prolonging $xg_1,\ldots,xg_{t-1}$ instead of $xg_1,\ldots,xg_{t}$, the size of the coefficient matrices is nice controlled.
This fact is more important in algorithm GHN$_n$.
Our second observation is that the complexity bound $O^{~}(d^{13.38}h^{2.004})$
in  Corollary \ref{cor-com1} is not reached in most cases and the
algorithm terminates in a much smaller number of loops.
So a further problem is to find a better complexity bound
or the average complexity for the algorithm.

\begin{figure}[H]
\centering
\includegraphics[scale=0.2]{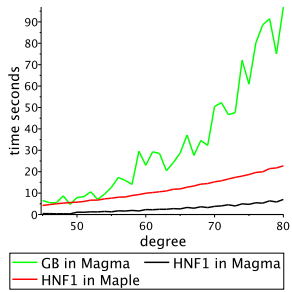}
\caption{Comparison of {\rm GHNF}$_1$ and Gr{\"o}bnerBasis in Magma and Maple: the $\Zx$ case}
\label{Z[x]eps}
\end{figure}

In Table \ref{table-1}, we give the timings for several input where the
polynomials have larger degrees. Other parameters are the same.
We see that for input polynomials with degree larger than 150,
the Gr{\"o}bnerBasis algorithm in Magma cannot compute in the GHNF in reasonable time.
The difference for the timings of Algorithm {\rm GHNF}$_1$  in Magma and Maple is mainly due to the different implementations of the HNF algorithms.

\begin{table}[H]
\caption{Comparison of {\rm GHNF}$_1$ and Gr{\"o}bnerBasis in Magma and Maple: the $\Z[x]$ case}\label{table-1}
\begin{center}
\begin{tabular}{|c|c|c|c|}
\hline
d & {\rm GHNF}$_1$ in Maple 18 & {\rm GHNF}$_1$ in Magma 2.21-7 & GB in Magma 2.21-7 \\
\hline
100 & 50.5932 & 19.048 &  214.91 \\
\hline
150 & 202.8135 & 104.827 & $>$1000  \\
\hline
200 & 590.7763 & 384.946 &  $>$1000 \\
\hline
\end{tabular}
\end{center}
\end{table}

Table \ref{Z[x]neps} plots the timings of Algorithm {\rm GHNF}$_n$ implemented in Magma 2.21-7 and Maple 18,
where the input random polynomial matrices are of size $3\times 3$ with degrees in $[2,30]$.
There is no implementation of Gr\"obner bases methods in Magma for $\Z[x]$-modules,
so we cannot make a comparison with Magma in this case.
In line with our complexity analysis given in Section 4,
algorithm {\rm GHNF}$_n$ slows down rapidly when $n>1$.

\begin{figure}[H]
\centering
\includegraphics[scale=0.2]{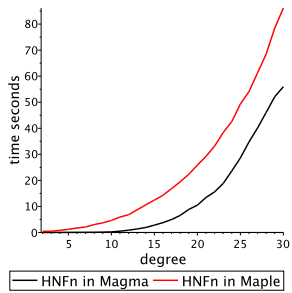}
\caption{Timings of {\rm GHNF}$_n$ in Magma and Maple}
\label{Z[x]neps}
\end{figure}

In Table \ref{table-n}, we list the timings of Algorithm {\rm GHNF}$_n$ for
several examples with larger degrees. This shows the polynomial-time
natural of the algorithm, because the algorithm works
for quite large $d$. Also, for large $d$, the Maple implementation
becomes faster.

\begin{table}[H]
\caption{Timings of {\rm GHNF}$_n$ in Magma and Maple}\label{table-n}
\begin{center}
\begin{tabular}{|c|c|c|}
\hline
d &{\rm GHNF}$_n$ in Maple 18 & {\rm GHNF}$_n$ in Magma 2.21-7  \\
\hline
40 &245.689
& 236.029\\
\hline
50 & 554.452
& 637.05  \\
\hline
\end{tabular}
\end{center}
\end{table}

\section{Conclusion}
In this paper, a polynomial-time algorithm is given to compute the \gHNFs of matrices over $\Zx$,
or equivalently, the reduced Gr\"obner basis of a $\Zx$-lattice.
The algorithm adopts the F4 strategy to compute Gr\"obner bases,
where a novel prolongation is designed so that the coefficient
matrices under consideration have smaller sizes than existing methods.
Existing efficient algorithms are used to compute the HNF for these coefficient
matrices.
Finally, nice degree and height bounds of elements of the reduced Gr\"obner basis are given.
The algorithm is implemented in Maple and Magma and is shown to be
more efficient than existing algorithms.

\section*{Acknowledgement}
We would like to thank Dr. Jianwei Li
for providing us information on the complexity
of computing Hermite normal forms.


\end{CJK*}
\end{document}